\documentclass[a4paper,12pt]{article}

\makeatletter
 
  \@addtoreset{equation}{section}
 \makeatother
\usepackage{amsfonts}
\usepackage{amssymb}
\usepackage{mathrsfs}
\usepackage{amsmath,amsthm}
\usepackage{amsmath}

\usepackage{graphicx}
\usepackage{color}
\usepackage{indentfirst}

\usepackage{multirow}

\begin{document}

\newtheorem{theorem}{Theorem}
\newtheorem{lemma}{Lemma}
\newtheorem{proposition}{Proposition}
\newtheorem{corollary}{Corollary}
\theoremstyle{definition}
\newtheorem{defn}{Definition}
\newtheorem{remark}{Remark}
\newtheorem{step}{Step}

\newcommand{\Cov}{\mathop {\rm Cov}}
\newcommand{\Var}{\mathop {\rm Var}}
\newcommand{\E}{\mathop {\rm E}}
\newcommand{\const }{\mathop {\rm const }}
\everymath {\displaystyle}

\newcommand{\ruby}[2]{
\leavevmode
\setbox0=\hbox{#1}
\setbox1=\hbox{\tiny #2}
\ifdim\wd0>\wd1 \dimen0=\wd0 \else \dimen0=\wd1 \fi
\hbox{
\kanjiskip=0pt plus 2fil
\xkanjiskip=0pt plus 2fil
\vbox{
\hbox to \dimen0{
\small \hfil#2\hfil}
\nointerlineskip
\hbox to \dimen0{\mathstrut\hfil#1\hfil}}}}

\everymath {\displaystyle}

\allowdisplaybreaks[4]

\def\qedsymbol{$\blacksquare$}
\renewcommand{\thefootnote }{\fnsymbol{footnote}}
\renewcommand{\refname }{References}
\renewcommand{\figurename}{Figure}

\title{Order Estimates for the Exact Lugannani--Rice Expansion}
\author{Takashi Kato
\footnote{Division of Mathematical Science for Social Systems, 
              Graduate School of Engineering Science, 
              Osaka University, 
              1-3, Machikaneyama-cho, Toyonaka, Osaka 560-8531, Japan, 
E-mail: \texttt{kato@sigmath.es.osaka-u.ac.jp}} \and 
Jun Sekine 
\footnote{Division of Mathematical Science for Social Systems, 
              Graduate School of Engineering Science, 
              Osaka University, 
              1-3, Machikaneyama-cho, Toyonaka, Osaka 560-8531, Japan, 
E-mail: \texttt{sekine@sigmath.es.osaka-u.ac.jp}} \and 
Kenichi Yoshikawa 
\footnote{Sumitomo Mitsui Banking Corporation, 
E-mail: \texttt{k.yoshi6208@gmail.com}}
}
\date{First version: October 12, 2013\\
This version: June 15, 2014}
\maketitle

\begin{abstract}
The Lugannani--Rice formula
is a saddlepoint approximation method 
for estimating the tail probability distribution function,
which was originally studied 
for the sum of independent identically distributed random variables.
Because of its tractability, 
the formula is now widely used 
in practical financial engineering 
as an approximation formula 
for the distribution of a (single) random variable.
In this paper, 
the Lugannani--Rice approximation formula is derived 
for a general, parametrized sequence 
$(X^{(\varepsilon)})_{\varepsilon>0}$ 
of random variables and the order estimates 
(as $\varepsilon\to 0$) 
of the approximation are given.

\footnote[0]{Mathematical Subject Classification (2010) \  62E17, 91G60, 65D15}\\
\footnote[0]{JEL Classification (2010) \ C63 , C65}
\end{abstract}

\section{Introduction}\label{sec_Intro}

Saddlepoint approximations (SPAs) provide effective methods 
for approximating
probability density functions and 
tail probability distribution functions, 
using their cumulant generating functions (CGFs).
In mathematical statistics, SPA methods originated with 
Daniels (1954), in which 
an approximation formula was given 
for the density function of the sample mean 
$\bar{X}_n = (X_1 + \cdots + X_n)/n$ of 
independent identically distributed (i.i.d.)\hspace{2mm}random variables $(X_i)_{i\in {\mathbb N}}$,
provided that the law of $X_1$ has the density function. 
Lugannani and Rice (1980) derives 
the following approximation formula 
for the right tail probability: 
\begin{equation}\label{original_LR}
P(\bar{X}_n > x) = 1-\Phi (\hat{w}_n) + 
\phi (\hat{w}_n)\left( \frac{1}{\hat{u}_n} - \frac{1}{\hat{w}_n}\right)  
+ O(n^{-3/2}) 
\end{equation}
as $n\to\infty$.
Here, $\Phi (w)$ and $\phi (w)$ are the standard normal distribution function 
and its density function $\phi:=\Phi'$, respectively, 
and $\hat{u}_n$ and $\hat{w}_n$ are expressed by using 
the CGF $K(\cdot )$ of $X_1$ and 
the saddlepoint $\hat{\theta }$ of $K(\cdot)$. 
That is,  $\hat{\theta}$ satisfies $K'(\hat{\theta }) = x$. 
Related SPA formulae have been studied in 
Daniels (1987), 
Jensen (1995), 
Kolassa (1997), 
Butler (2007), 
the references therein, and others. 

Strictly, the Lugannani--Rice (LR) formula (\ref {original_LR})
should be interpreted as an asymptotic result as $n\rightarrow \infty$.
However, it is popular in many practical applications of financial engineering 
as an approximation formula for the right tail probability because of its tractability.
This approximation is
\begin{eqnarray}\label{0th_LR}
P(X_1 > x) \approx   1-\Phi (\hat{w}_1) + 
\phi (\hat{w}_1)\left( \frac{1}{\hat{u}_1} - \frac{1}{\hat{w}_1}\right) . 
\end{eqnarray}

In other words, 
LR formula (\ref {original_LR}) 
is applied even when $n$ is $1$ !
For financial applications of SPA formulae, 
we refer the readers to papers such as Rogers and Zane (1999), 
Xiong, Wong, and Salopek (2005),
A\"it-Sahalia and Yu (2006), 
Yang, Hurd, and Zhang (2006),
Glasserman and Kim (2009),
and Carr and Madan (2009). 
It is interesting that
the approximation formula (\ref {0th_LR}) 
still works surprisingly well in many financial examples, 
despite its lack of theoretical justification.

The aim of this paper is to provide a measure of the effectivity 
of the ``generalized usage'' of the LR formula (\ref {0th_LR})
from an asymptotic theoretical viewpoint.
We consider a general parametrized sequence of random variables
$(X^{(\varepsilon)})_{\varepsilon> 0}$ 
and assume that the $r$th cumulant of $X^{(\varepsilon)}$
has order $O(\varepsilon^{r-2})$ 
as $\varepsilon \rightarrow 0$ 
for each $r\geq 3$. 
This implies that 
$X^{(\varepsilon)}$ 
converges in law to 
a normally distributed random variable 
(a motivation is provided for this assumption in Remark \ref {rem_A5} of Section \ref {sec_order}).
We next derive the expansion
\begin{equation}\label{exact_LR}
P\left(X^{(\varepsilon)}>x \right) 
= 1 - \Phi (\hat{w}_\varepsilon ) 
+ \sum ^\infty _{m = 0}\Psi ^\varepsilon _m(\hat{w}_\varepsilon ),
\end{equation}
which we call the exact LR expansion (see Theorem 1 of Section 2). 
Here, $\hat{w}_\varepsilon$ is given by 
(\ref {def_SP}) and (\ref {def_w_eps}), 
and the
$\Psi ^\varepsilon _m(\hat{w}_\varepsilon )$ 
($m\in {\mathbb Z}_+$) are given by 
(\ref {def_Psi_m}).
We then show that
\begin{equation}\label{order_01}
\text{
$\Psi ^\varepsilon _0(\hat{w}_\varepsilon ) = O(\varepsilon )$ 
and $\Psi ^\varepsilon _m(\hat{w}_\varepsilon ) = O(\varepsilon ^3)$
as $\varepsilon\rightarrow 0$ for all $m\in {\mathbb N}$ }
\end{equation}
under some conditions. 
This is the main result of the paper
(see Theorem \ref {th_main} in Section \ref {sec_order} for the details). 
\begin{remark}\label{rem_intro}
We note that the expansion (\ref {exact_LR}) 
with the order estimates (\ref {order_01})
and the classical LR formula (\ref {original_LR}) 
treat different situations, 
although they may have some overlap. 
Let
\begin{eqnarray*}
\varepsilon:=\frac{1}{\sqrt{N}}
\quad\text{and}\quad
X^{(\varepsilon)}:=\varepsilon \sum_{i=1}^{1/\varepsilon^2} X_i,  
\end{eqnarray*}
where $(X_i)_{i\in {\mathbb N}}$ is an i.i.d.\hspace{1mm}sequence of random variables.
Then, we can check that the law of  $X^{(\varepsilon)}$
satisfies the conditions
necessary to apply Theorem \ref {th_main} in Section \ref {sec_order} 
(see Remark \ref {rem_A5} (iv) in Section \ref {sec_order}).
So, (\ref {exact_LR}) holds with (\ref {order_01}).
On the other hand, the classical LR formula (\ref {original_LR}) 
gives an approximation formula of the {\it far}-right tail probability: 
\begin{eqnarray*}
P\left({X}^{(\varepsilon)} > \frac{x}{\varepsilon} \right) 
= 1-\Phi (\hat{w}_\varepsilon) + 
\phi (\hat{w}_\varepsilon)
\left( \frac{1}{\hat{u}_\varepsilon} - \frac{1}{\hat{w}_\varepsilon}\right)  
+ O(\varepsilon^{3}) \ \ \mbox {as} \ \ \varepsilon \rightarrow 0. 
\end{eqnarray*}
In this paper, 
with motivation from financial applications
(e.g., call option pricing in Section \ref {sec_eg}), 
we choose to analyse the right tail probability $P(X^{(\varepsilon)}>x)$ 
instead of the far-right tail probability 
$P(X^{(\varepsilon)}>x/\varepsilon)$.
For a related remark, see (i) in Section \ref {sec_conclusion}. 
\end{remark}
The organisation of the rest of this paper is as follows. 
In Section \ref {sec_LR}, 
we introduce the ``exact'' LR expansion: 
we first derive it formally, 
and next provide a technical condition sufficient 
to ensure the validity of the expansion.
Section \ref {sec_order} states our main results:
we derive the order estimates of the higher order terms 
in the exact LR expansion (\ref {exact_LR}).
Section \ref {sec_eg} discusses some examples: 
we introduce two stochastic volatility (SV) models and 
numerically check the accuracy of the higher order LR formula. 
Section \ref {sec_proofs} contains the necessary proofs:
Subsection \ref {sec_proof_LR} gives the proof of Theorem \ref {th_exact_LR} 
and Subsection \ref {sec_proof_estimate} gives the proof of Theorem \ref {th_main}.
Section \ref {sec_extention} discusses some extensions of Theorem \ref {th_main}:
under additional conditions 
we obtain the sharper estimate 
$\Psi ^\varepsilon _m(\hat{w}_\varepsilon ) = O(\varepsilon ^{2m + 1})$ as 
$\varepsilon \rightarrow 0$ 
for $m\in {\mathbb N}$,  
and the related order estimate 
of the absolute error of the $M$th order LR formula.
In addition, we introduce error estimates for the Daniels-type formula,
which is an approximation formula for 
the probability density function.
The last Section \ref {sec_conclusion} contains concluding remarks. 
In Appendix, we present some
toolkits for deriving 
the explicit forms of 
$\Psi ^\varepsilon _2(\hat{w}_\varepsilon )$ and $\Psi ^\varepsilon _3(\hat{w}_\varepsilon )$.

\section{The Exact Lugannani--Rice Expansion}\label{sec_LR}

In this section we derive the exact LR expansion (\ref {exact_LR}), 
which is given as a natural generalisation of the original LR formula. 
For readability, we introduce here the formal calculations 
to derive that formula and leave 
rigorous arguments to Section \ref {sec_proof_LR} 
(see also Appendix in Rogers and Zane (1999)).

Let $(\mu _\varepsilon )_{0\leq \varepsilon \leq 1}$ 
be a family of probability distribution on $\Bbb {R}$ and 
define a distribution function $F_\varepsilon $ and a tail probability function $\bar{F}_\varepsilon $ by 
\begin{eqnarray*}
F_\varepsilon (x) = \mu _\varepsilon ((-\infty , x]), \ \ \bar{F}_\varepsilon (x) = 1 - F_\varepsilon (x). 
\end{eqnarray*}
We denote by $K_\varepsilon $ the CGF of $\mu _\varepsilon $, that is, 
\begin{eqnarray*}
K_\varepsilon (\theta ) = \log \int _\Bbb {R}e^{\theta x}\mu _\varepsilon (dx). 
\end{eqnarray*}
We assume the following conditions. 
\begin{itemize}
 \item [ {[A1]} ] 
For each $\varepsilon \in [0, 1]$, the effective domain 
$\mathcal {D}_\varepsilon = \{ \theta \in \Bbb {R}\ ; \ |K_\varepsilon (\theta )| < \infty  \}$ of $K_\varepsilon $ 
contains an open interval that includes zero. 
 \item [ {[A2]} ] 
For each $\varepsilon \in [0, 1]$, 
the support of $\mu _\varepsilon $ is equal to the whole line $\Bbb {R}$. 
Moreover, the characteristic function of $\mu _\varepsilon $ is integrable; that is, 
\begin{eqnarray*}
\int ^\infty _{-\infty }\left| \int ^\infty _{-\infty }e^{i\xi x}\mu _\varepsilon (dx)\right| d\xi < \infty , 
\end{eqnarray*}
where $i = \sqrt{-1}$ is the imaginary unit. 
\end{itemize}

It is well known that $K_\varepsilon $ is analytic and convex on the interior 
$\mathcal {O}_\varepsilon $ of $\mathcal {D}_\varepsilon $. 
Moreover, [A2] implies that $\mu _\varepsilon $ has a density function, and thus 
$K_\varepsilon $ is a strictly convex function (see Durrett (2010), for instance). 
Since the range of $K'_\varepsilon $ coincides with $\Bbb {R}$ under [A1]--[A2], 
we can always find the solution $\hat{\theta }_\varepsilon = \hat{\theta }_\varepsilon (x) \in \mathcal {O}_\varepsilon $ to 
\begin{eqnarray}\label{def_SP}
K'_\varepsilon (\hat{\theta }_\varepsilon ) = x
\end{eqnarray}
for any $x\in \Bbb {R}$. 
We call $\hat{\theta }_\varepsilon $ the saddlepoint of $K_\varepsilon $ given $x$. 
Here, note that $K_\varepsilon $ is analytically continued 
as the function defined on $\mathcal {O}_\varepsilon \times i\Bbb {R}$. 

Now, we derive (\ref {exact_LR}). 
Until the end of this section, we fix an $\varepsilon \in [0, 1]$ and an $x\in \Bbb {R}$. 
To derive (\ref {exact_LR}), we further that require the condition $\hat{\theta }_\varepsilon \neq 0$ be satisfied. 
Applying Levy's inversion formula, 
we represent $\bar{F}_\varepsilon (x)$ by the integral form 
\begin{eqnarray}\label{calc_tail}
\bar{F}_\varepsilon (x) 
\ = \ 
\frac{1}{2\pi i}\int ^{c + i\infty }_{c - i\infty } \exp (K_\varepsilon (\theta ) - x\theta )\frac{d\theta }{\theta } 
\label{Levy}
\end{eqnarray}
for arbitrary $c\in \mathcal {O}_\varepsilon \setminus \{0\}$ 
(see Proposition \ref {prop_inversion} in Subsection \ref {sec_proof_LR}). 

Next, we represent $\hat{w}_\varepsilon \in \Bbb {R}$ as 
\begin{eqnarray}\label{def_w_eps}
\hat{w}_\varepsilon  = \mathrm {sgn}(\hat{\theta }_\varepsilon )
\sqrt{2(x
\hat{\theta }_\varepsilon  - K_\varepsilon (\hat{\theta }_\varepsilon ))}, 
\end{eqnarray}
where $\mathrm {sgn }(a) = 1 \ (a \geq 0), \ -1 \ (a < 0)$. 
Note that $\hat{w}_\varepsilon $ is well defined because of the calculation 
\begin{eqnarray}\nonumber 
x\hat{\theta }_\varepsilon  - K_\varepsilon (\hat{\theta }_\varepsilon ) &=& 
K_\varepsilon (0) - K_\varepsilon (\hat{\theta }_\varepsilon ) + 
K'_\varepsilon (\hat{\theta }_\varepsilon )\hat{\theta }_\varepsilon \\
&=& 
\int ^1_0(1-u)K''_\varepsilon (-u\hat{\theta }_\varepsilon )du\hat{\theta }_\varepsilon ^2 \ \geq \ 0
\label{calc_K_diff}
\end{eqnarray}
by virtue of the convexity of $K_\varepsilon $ and Taylor's theorem. 
We consider the following change of variables between $w$ and $\theta $: 
\begin{eqnarray}\label{change_variable}
\frac{1}{2}w^2 - \hat{w}_\varepsilon w = K_\varepsilon (\theta ) - x\theta. 
\end{eqnarray}
Then, replacing the variable $\theta $ with $w$ in the right-hand side of (\ref {calc_tail}) 
and applying Cauchy's integral theorem, we see that 
\begin{eqnarray}\nonumber 
\bar{F}_\varepsilon (x) &=& 
\frac{1}{2\pi i}\int _{\gamma _\varepsilon }
\exp \left( \frac{1}{2}w^2 - \hat{w}_\varepsilon w\right) 
\frac{\theta '(w)}{\theta (w)}dw\\ 
&=& 
\frac{1}{2\pi i}\int ^{\hat{w}_\varepsilon + i\infty }_{\hat{w}_\varepsilon - i\infty }
\exp \left( \frac{1}{2}w^2 - \hat{w}_\varepsilon w\right) 
\frac{\theta '(w)}{\theta (w)}dw, 
\label{Cauchy_bar_F}
\end{eqnarray}
where 
$\gamma _\varepsilon $ is a Jordan curve in $w$-space 
corresponding to the line 
$\{ \hat{\theta }_\varepsilon \} \times i\Bbb {R}$ 
and $\theta (w) \ (= \theta _\varepsilon (w))$ is defined by (\ref {change_variable}) as an implicit function with respect to $w$. 
Note that $\theta (w)$ is well defined for each $w$ and is analytic on each contour 
under suitable conditions. 
Denoting 
\begin{eqnarray*}
\psi _\varepsilon (w) = \frac{\theta '(w)}{\theta (w)} - \frac{1}{w} = \frac{d}{dw}\log \left( \frac{\theta (w)}{w}\right), 
\end{eqnarray*}
we can decompose (\ref {Cauchy_bar_F}) into 
\begin{eqnarray*}
\bar{F}_\varepsilon (x) = N_\varepsilon (x) + 
\frac{1}{2\pi i}\int ^{\hat{w}_\varepsilon + i\infty }_{\hat{w}_\varepsilon - i\infty }
\exp \left( \frac{1}{2}w^2 - \hat{w}_\varepsilon w\right) \psi _\varepsilon (w)dw, 
\end{eqnarray*}
where 
\begin{eqnarray*}
N_\varepsilon (x) &=& 
\frac{1}{2\pi i}\int ^{\hat{w}_\varepsilon + i\infty }_{\hat{w}_\varepsilon - i\infty }
\exp \left( \frac{1}{2}w^2 - \hat{w}_\varepsilon w\right) \frac{dw}{w}. 
\end{eqnarray*}
$N_\varepsilon (x)$ is just the tail probability of the standard normal distribution; that is, 
$N_\varepsilon (x) = \bar{\Phi }(\hat{w}_\varepsilon )$, 
where 
\begin{eqnarray*}
\bar{\Phi }(w)\ = \ \int ^\infty _w\phi (y)dy, \ \ 
\phi (y)\ = \ \frac{1}{\sqrt{2\pi }}e^{-y^2/2}. 
\end{eqnarray*}
Here, if $\hat{w}_\varepsilon \neq 0$, we see that 
$\psi _\varepsilon $ is analytic on $\{ \hat{w}_\varepsilon \} \times  i\Bbb {R}$; 
hence, we obtain 
\begin{eqnarray}\nonumber 
&&
\frac{1}{2\pi i}\int ^{\hat{w}_\varepsilon + i\infty }_{\hat{w}_\varepsilon - i\infty }
\exp \left( \frac{1}{2}w^2 - \hat{w}_\varepsilon w\right) \psi _\varepsilon (w)dw\\\nonumber 
&=& 
\frac{1}{2\pi i}\int ^{\hat{w}_\varepsilon + i\infty }_{\hat{w}_\varepsilon - i\infty }
\exp \left( \frac{1}{2}w^2 - \hat{w}_\varepsilon w\right)
\sum ^\infty _{n = 0}\frac{\psi _\varepsilon ^{(n)}(\hat{w}_\varepsilon )}{n!}(w - \hat{w}_\varepsilon )^ndw\\\nonumber 
&=& 
\frac{1}{2\pi }e^{-\hat{w}_\varepsilon ^2/2}\int ^{\infty }_{-\infty }
e^{-y^2/2}
\sum ^\infty _{n = 0}\frac{\psi _\varepsilon ^{(n)}(\hat{w}_\varepsilon )}{n!}(iy)^ndy\\
&=& 
\frac{1}{2\pi }e^{-\hat{w}_\varepsilon ^2/2}\sum ^\infty _{n = 0}
i^n\psi _\varepsilon ^{(n)}(\hat{w}_\varepsilon )\int ^\infty _{-\infty }
e^{-y^2/2}
\frac{y^n}{n!}dy\ = \ 
\sum ^\infty _{m = 0}\Psi ^\varepsilon _m(\hat{w}_\varepsilon ), 
\label{integral_sum}
\end{eqnarray}
where we define
\begin{equation}\label{def_Psi_m}
\Psi_m^{\varepsilon}(w) 
= \phi (w)\frac{(-1)^m}{(2m)!!}\psi ^{(2m)}_{\varepsilon}(w) 
= \phi (w)\frac{(-1)^m}{(2m)(2m-2)\cdots 4\cdot 2}\psi ^{(2m)}_{\varepsilon}(w).
\end{equation}
This is the exact LR expansion (\ref {exact_LR}). 
Note here that the $0$th order approximation formula
\begin{eqnarray*}
\bar{\Phi }(\hat{w}_\varepsilon ) + \Psi _0^{\varepsilon}
(\hat{w}_\varepsilon )  
\end{eqnarray*}
corresponds to the original LR formula (\ref {original_LR}). 
Indeed, we see that
\begin{eqnarray*}
\Psi _0^{\varepsilon}(\hat{w}_\varepsilon ) 
= \phi (\hat{w}_\varepsilon )\left\{ \frac{1}{\hat{\theta }_\varepsilon \sqrt{K''_\varepsilon (\hat{\theta }_\varepsilon )}} - \frac{1}{\hat{w}_\varepsilon }\right\}.
\end{eqnarray*}
The $1$st order approximation formula 
\begin{eqnarray*}
\bar{\Phi }(\hat{w}_\varepsilon ) 
+ \Psi ^{\varepsilon}_0(\hat{w}_\varepsilon ) 
+ \Psi ^{\varepsilon}_1(\hat{w}_\varepsilon ) 
\end{eqnarray*}
is also often called the LR formula, where we have that
\begin{eqnarray*}
\Psi _1^{\varepsilon}(\hat{w}_\varepsilon ) = 
\phi (\hat{w}_\varepsilon )\left\{ \frac{1}{\hat{\theta }_\varepsilon 
\sqrt{K''_\varepsilon (\hat{\theta }_\varepsilon )}}
\left( \frac{1}{8}\hat{\lambda }_4 
- \frac{5}{24}\hat{\lambda }^2_3\right) 
- \frac{1}{2\hat{\theta }_\varepsilon ^2K''_\varepsilon 
(\hat{\theta }_\varepsilon )}\hat{\lambda }_3 - 
\left( \frac{1}{\hat{\theta }_\varepsilon ^3(
K''_\varepsilon (\hat{\theta }_\varepsilon ))^{3/2}} 
- \frac{1}{\hat{w}_\varepsilon ^3}\right) \right\}
\end{eqnarray*}
with
\begin{eqnarray*}
\hat{\lambda }_3 = \frac{K^{(3)}_\varepsilon (\hat{\theta }_\varepsilon )}{\left( K''_\varepsilon (\hat{\theta }_\varepsilon )\right) ^{3/2}}, \ \ 
\hat{\lambda }_4 = \frac{K^{(4)}_\varepsilon (\hat{\theta }_\varepsilon )}{\left( K''_\varepsilon (\hat{\theta }_\varepsilon )\right) ^2}. 
\end{eqnarray*}
The explicit forms of the higher order terms  
$\Psi _2^\varepsilon(\hat{w}_\varepsilon )$ 
and $\Psi _3^\varepsilon(\hat{w}_\varepsilon )$ 
are shown in Appendix. 

The above formal derivation of the exact LR expansion (\ref {exact_LR})
can be made rigorous under suitable conditions, such as the following. 
\begin{itemize}
 \item [ {[B1]} ] For each $\varepsilon \in [0, 1]$, 
there exists $\delta _\varepsilon , C_\varepsilon > 0$ such that 
$\delta _\varepsilon \leq |K''_\varepsilon | \leq C_\varepsilon $ on 
$\mathcal {O}_\varepsilon \times i\Bbb {R}$. 
 \item [ {[B2]} ] 
The range of the holomorphic map 
$\iota _\varepsilon : \mathcal {O}_\varepsilon \times i\Bbb {R} \longrightarrow \Bbb {C}$ defined by 
\begin{eqnarray*}
\iota _\varepsilon (\theta ) = 
K_\varepsilon (\theta ) - x\theta - 
(K_\varepsilon (\hat{\theta }_\varepsilon ) - x\hat{\theta }_\varepsilon ) 
\end{eqnarray*} 
includes a convex set that contains 
$\{ 2\iota _\varepsilon (\hat{\theta }_\varepsilon  + it)\ ; \ t\in \Bbb {R}\} $ and 
$(-\infty , 0]$, 
 \item [ {[B3]} ] $\sum ^\infty _{n=1}
|\psi _\varepsilon ^{(n)}(\hat{w}_\varepsilon )|/(n!!) < \infty $. 
\end{itemize}
Under these conditions, we obtain the following,  
whose proof is given in Subsection \ref {sec_proof_LR}.
\begin{theorem}\label{th_exact_LR} \ 
Assume $\mathrm {[A1]}$--$\mathrm {[A2]}$ and 
$\mathrm {[B1]}$--$\mathrm {[B3]}$. Then $(\ref {exact_LR})$ holds. 
\end{theorem}

\section{Order Estimates of Approximation Terms}\label{sec_order}

In practical applications, we need to truncate the formula (\ref {exact_LR}) 
with $M\in {\mathbb N}$
\begin{eqnarray}\label{LR_formula_sim}
\bar{F}_\varepsilon (x) \approx
\bar{F}^M_\varepsilon (x) := 
\bar{\Phi }(\hat{w}_\varepsilon ) 
+ \sum ^{M}_{m = 0}\Psi ^\varepsilon _m(\hat{w}_\varepsilon ). 
\end{eqnarray}
We call the right-hand side of (\ref {LR_formula_sim}) 
the $M$th LR formula. 
The aim of this section is to derive order estimates for 
$\Psi _m(\hat{w}_\varepsilon )$ ($m = 0, 1, \ldots , $) 
as $\varepsilon \rightarrow 0$. 

We fix $x\in \Bbb {R}$, which is an arbitrary value such that
\begin{equation}\label{cond_not_x}
\int _\Bbb {R}y\mu _0(dy) \neq x. 
\end{equation}
We then impose the following additional assumptions.
\begin{itemize}
 \item [ {[A3]} ] There is a $\delta _0 > 0$ such that $K''_\varepsilon (\theta ) \geq \delta _0$ for each 
$\theta \in \mathcal {O}_\varepsilon $ and $\varepsilon \in [0, 1]$. 
 \item [ {[A4]} ] For each $\varepsilon $, there is an interval 
$\mathcal {I}_\varepsilon \subset \mathcal {D}_\varepsilon $ such that 
$\mathcal {I}_\varepsilon \nearrow \Bbb {R}$ as $\varepsilon \rightarrow 0$; that is, 
$\mathcal {I}_\varepsilon \subset \mathcal {I}_{\varepsilon '}$ for each $\varepsilon \geq \varepsilon '$ and 
$\cup _{\varepsilon }\mathcal {I}_\varepsilon = \Bbb {R}$. 
 \item [ {[A5]} ] For each nonnegative integer $r$, 
$K^{(r)}_\varepsilon (\theta )$ 
converges uniformly to $K^{(r)}_0(\theta )$ with $\varepsilon \rightarrow 0$ 
on any compact subset of $\Bbb {R}$. 
Moreover, for each integer $r\geq 3$, 
$K^{(r)}_\varepsilon (\theta )$ has order $O(\varepsilon ^{r - 2 })$ as 
$\varepsilon \rightarrow 0$ in the following sense: 
For each compact set $C\subset \Bbb {R}$, it holds that 
\begin{eqnarray}\label{cond_order}
\limsup _{\varepsilon \rightarrow 0}\sup _{\theta \in C}\varepsilon ^{-(r - 2)}|K^{(r)}_\varepsilon (\theta )| < \infty. 
\end{eqnarray}
\end{itemize}

\begin{remark}\label{rem_A5} \ 
\begin{itemize}
 \item [ $\mathrm {(i)}$ ] To derive the formula $(\ref {exact_LR})$, 
we need that $\hat{\theta }_\varepsilon \neq 0$. 
This condition is satisfied for small $\varepsilon $ under $(\ref {cond_not_x})$ and $\mathrm {[A5]}$ jointly. 
See Corollary \ref {cor_not_zero} in Section \ref {sec_proof_estimate} for the details. 

 \item [ $\mathrm {(ii)}$ ] From $\mathrm {[A4]}$, 
we see that for each compact set $C\subset \Bbb {R}$ there is an $\varepsilon _0$ such that 
$C\subset \mathcal {D}_{\varepsilon }$ for $\varepsilon \leq \varepsilon _0$. 
Therefore, the assertions in $\mathrm {[A5]}$ make sense for small $\varepsilon $. 
Note that one of the sufficient conditions for $\mathrm {[A4]}$ is that
\begin{itemize}
 \item [ {[A4']} ] $\mathcal {D}_\varepsilon \nearrow \Bbb {R}$, \ $\varepsilon \rightarrow 0$. 
\end{itemize}

 \item [ $\mathrm {(iii)}$ ] $\mathrm {[A5]}$ implies that 
$K^{(r)}_0(\theta ) = 0$ holds for $r\geq 3$. Therefore, 
\begin{eqnarray*}
K_0(\theta ) = m\theta + \frac{1}{2}\sigma ^2\theta ^2
\end{eqnarray*}
with some $m\in {\mathbb R}$ 
and $\sigma > 0$, 
where the positivity of $\sigma$ follows from
([A2] or) [A3].
Hence, $\mu _0$ is the normal distribution 
with mean $m$ and variance $\sigma ^2$. 
Note here that the effective domain of $K_0$ 
is equal to $\Bbb {R}$, which is consistent with [A4]. 

 \item [ $\mathrm {(iv)}$ ] An example which satisfies
$\mathrm {[A5]}$ is the following.
Let $X_i$ for $i\in \Bbb {N}$ be i.i.d.\hspace{1mm}random variables 
with mean zero, let $\tilde{X}_n = (X_1 + \cdots X_n)/\sqrt{n}$,
and let $\mu _{1/\sqrt{n}}$ be its distribution.
We see that $(\mu _{1/\sqrt{n}})_{n}$ satisfies [A5]
by the central limit theorem 
(setting $\varepsilon:=1/\sqrt{n}$).
SV models with small ``vol of vol'' parameters 
are introduced as additional examples in Section \ref {sec_eg}.
\end{itemize}
\end{remark}

Now, we introduce our main theorem. 
\begin{theorem}\label{th_main} \ 
Assume that conditions $\mathrm {[A1]}$--$\mathrm {[A5]}$ hold. 
Then $\Psi ^\varepsilon _0(\hat{w}_\varepsilon ) = O(\varepsilon )$ 
and $\Psi ^\varepsilon _m(\hat{w}_\varepsilon ) = O(\varepsilon ^3)$, 
both as $\varepsilon\to 0$
for each $m\geq 1$. 
\end{theorem}
Recall here that the notation 
$a_\varepsilon = O(\varepsilon ^r)$ implies 
$\limsup _{\varepsilon \rightarrow 0}
\varepsilon ^{-r}|a_\varepsilon | < \infty$. 
\begin{remark}\label{rem_order} \ 
It may be natural to expect that
$\Psi _m^\varepsilon(\hat{w}_\varepsilon ) = O(\varepsilon ^{k_m})$ holds 
as $\varepsilon \rightarrow 0$ for some $k_m > 3$. 
In other words, to expect that the relation 
$\Psi _m^\varepsilon(\hat{w}_\varepsilon ) = \Theta (\varepsilon ^3)$ 
may not hold for $m\geq 2$. 
Here, $a_n = \Theta (b_n)$ 
is the Bachmann--Landau ``Big-Theta'' notation, 
meaning that
\begin{eqnarray*}
0 < \liminf _n \frac{a_n}{b_n} \leq \limsup _n \frac{a_n}{b_n} < \infty. 
\end{eqnarray*}
Under conditions [A1]--[A5], 
we have not obtained sharper estimates 
for $\Psi _m^\varepsilon(\hat{w}_\varepsilon )$  ($m\geq 2$) 
than given in Theorem \ref {th_main}.
In Section \ref {sec_error_estimate} we show that by assuming [A6]--[A7] additionally we obtain 
\begin{eqnarray}\label{higher_order_estimate}
\Psi _m^\varepsilon(\hat{w}_\varepsilon ) = O(\varepsilon ^{2m + 1})
\quad\text{as}\quad 
\varepsilon \rightarrow 0 
\end{eqnarray}
for each $m\geq 0$, and 
\begin{eqnarray}\label{error_estimate}
\bar{F}_\varepsilon (x) = \bar {\Phi }(\hat{w}_\varepsilon ) + 
\sum ^M_{m  = 0}\Psi _m^\varepsilon(\hat{w}_\varepsilon ) 
+ O(\varepsilon ^{2M + 3}) 
\quad\text{as}\quad 
\varepsilon \rightarrow 0 
\end{eqnarray}
for each $M\geq 0$. 
In the next section, 
we also numerically demonstrate these results by use of examples.
\end{remark}

\section{Examples}\label{sec_eg}

In this section, 
we introduce some examples and apply our results.

\subsection{The Heston SV model}\label{sec_Heston}

As the first example, 
we treat Heston's SV model (Heston (1993)). 
We consider the following stochastic differential equation (SDE): 
\begin{eqnarray*}
&&dX^\varepsilon _t = -\frac{1}{2}V^\varepsilon _tdt + \sqrt{V^\varepsilon _t}dB^1_t, \\
&&dV^\varepsilon _t = \kappa (b - V^\varepsilon _t)dt + \varepsilon \sqrt{V^\varepsilon _t}(\rho dB^1_t + \sqrt{1-\rho ^2}dB^2_t), \\
&&X^\varepsilon _0 = x_0, \ V^\varepsilon _0 = v_0, 
\end{eqnarray*}
where $\kappa , b > 0$, $\rho \in [-1, 1]$, and $\varepsilon \geq 0$. 
It is known that the above SDE has the unique solution $(X^\varepsilon _t, V^\varepsilon _t)_t$ 
when $2\kappa b\geq \varepsilon ^2$. 
The process $(X^\varepsilon _t)_t$ is regarded as the log-price process of a risky asset 
with the stocastic volatility process $(\sqrt{V^\varepsilon _t})_t$ 
under the risk-neutral probability measure (the risk-free rate is set as zero for simplicity). 
Our goal is to approximate the tail probability $\bar{F}_\varepsilon (x) = P(X^\varepsilon _T > x)$ 
for a time $T > 0$. 

Here $\varepsilon \geq 0$ is the ``vol of vol'' parameter, 
which describes dispersion of the volatility process. 
In this section, we consider the case of a small $\varepsilon $. 
Note that when $\varepsilon  = 0$, $X^\varepsilon _T$ has the normal distribution. 

To apply our main result, 
we verify that the conditions [A1]--[A5] for $\mu _\varepsilon = P(X^\varepsilon _T\in \cdot )$ hold. 
First, [A1] is satisfied and 
the explicit form of the CGF of $\mu _\varepsilon $ with $\varepsilon > 0$ is given as 
\begin{eqnarray}\label{CGF_Heston}
K_\varepsilon (\theta ) = 
x_0\theta + \frac{2\kappa b}{\varepsilon ^2}\left\{ 
\frac{1}{2}(\kappa - \varepsilon \rho \theta ) - \log q_\varepsilon (\theta )\right\} - 
\frac{v_0(\theta - \theta ^2)\sinh (\sqrt{p_\varepsilon (\theta )}T/2)}
{\sqrt{p_\varepsilon (\theta )}q_\varepsilon (\theta )} 
\end{eqnarray}
on a neighbourhood of the origin, where 
\begin{eqnarray*}
p_\varepsilon (\theta ) &=& (\kappa - \varepsilon \rho \theta )^2 + \varepsilon ^2(\theta - \theta ^2), \\
q_\varepsilon (\theta ) &=& \cosh \frac{\sqrt{p_\varepsilon (\theta )}T}{2} + 
\frac{\kappa - \varepsilon \rho \theta }{\sqrt{p_\varepsilon (\theta )}}
\sinh \frac{\sqrt{p_\varepsilon (\theta )}T}{2} 
\end{eqnarray*} 
(see Rollin, Castilla, and Utzet (2010) or Yoshikawa (2013)). 
Note that when $\varepsilon = 0$, we have 
\begin{eqnarray*}
K_0(\theta ) = \frac{1}{2}\sigma ^2(\theta ^2 - \theta ) + x_0\theta, 
\end{eqnarray*}
where $\sigma ^2 = bT + (v_0 - b)(1-e^{-\kappa T})/\kappa $. 
Moreover, Theorem 3.3 and Corollary 3.4 in Rollin, Castilla, and Utzet (2010) imply [A2]. 
The same source also tells us that 
when $\varepsilon \rho  < \kappa $, it is also true that 
$\mathcal {D}_\varepsilon \supset \mathcal {I}_\varepsilon  := [u_{\varepsilon , -}, u_{\varepsilon , +}]$, where 
$u_{\varepsilon , -} < 0 < u_{\varepsilon , +}$ are given by 
\begin{eqnarray*}
u_{\varepsilon , \pm} = \frac{\varepsilon - 2\kappa \rho \pm  
\sqrt{4\kappa ^2 + \varepsilon ^2 - 4\kappa \rho \varepsilon }}
{2\varepsilon (1-\rho ^2)}. 
\end{eqnarray*}
We can easily see that $u_{\varepsilon , +}\nearrow \infty $ and $u_{\varepsilon , -}\searrow -\infty $ as $\varepsilon \downarrow 0$; 
thus, [A4] is satisfied. 
[A5] is obtained by Theorem 3.1 in Yoshikawa (2013). 
Finally, we numerically compute the minimum value of $K''_\varepsilon $ for each $\varepsilon $ to confirm [A3]. 
We set the parameters as 
$\kappa = 1$, $b = 1$, $x_0 = 0$, $v_0 = 1$, $\rho = 0.3$, $T = 1$, and $x = 1$. 
Then we get Figure \ref {graph_min_KD2}, which implies that [A3] holds.

\begin{figure}[t]
\begin{center}
\includegraphics[width=10cm]{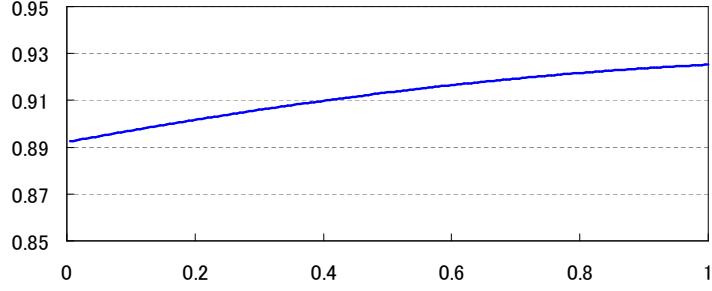}
\caption{Plots of $\inf _{\theta \in \mathcal {D}_\varepsilon }K''_\varepsilon (\theta )$. 
The horizontal axis corresponds to $\varepsilon $. }
\label{graph_min_KD2}
\end{center}
\end{figure}

\begin{remark}
Theorem 3.1 in Rollin, Castilla, and Utzet (2010) presents a method to calculate 
the lower bound $\theta ^*_{\varepsilon , -}$ and the upper bound $\theta ^*_{\varepsilon , +}$ of 
the effective domain $\mathcal {D}_\varepsilon $. 
When we set the parameters as above, 
the bounds $\theta ^*_{\varepsilon , \pm }$ are obtained by 
\begin{eqnarray*}
\theta ^*_{\varepsilon , +} &=& 
\mathrm {argmin}\{ \tilde{q}_\varepsilon (\theta )\ ; \ \theta \in (u_{\varepsilon , +}, \alpha _{\varepsilon , +1}) \} , \\
\theta ^*_{\varepsilon , +} &=& 
\mathrm {argmax}\{ \tilde{q}_\varepsilon (\theta )\ ; \ 
\theta \in (\alpha _{\varepsilon , -1}, u_{\varepsilon , -}) \} , 
\end{eqnarray*}
where 
\begin{eqnarray*}
\tilde{q}_\varepsilon (\theta ) = 
\cos \frac{\sqrt{-p_\varepsilon (\theta )}T}{2} + 
\frac{\kappa - \varepsilon \rho \theta }{\sqrt{-p_\varepsilon (\theta )}}
\sin \frac{\sqrt{-p_\varepsilon (\theta )}T}{2} 
\end{eqnarray*}
and $\alpha _{\varepsilon , -1} < 0 < \alpha _{\varepsilon , +1}$ are the solutions to 
$p_\varepsilon (\theta ) = -4\pi ^2/T^2$. 
Note that $K_\varepsilon (\theta )$ is given by (\ref {CGF_Heston}) 
on $[u_{\varepsilon , -}, u_{\varepsilon , +}]$ and by 
\begin{eqnarray*}
K_\varepsilon (\theta ) = 
x_0\theta + \frac{2\kappa b}{\varepsilon ^2}\left\{ 
\frac{1}{2}(\kappa - \varepsilon \rho \theta ) - \log \tilde{q}_\varepsilon (\theta )\right\} - 
\frac{v_0(\theta - \theta ^2)\sin (\sqrt{-p_\varepsilon (\theta )}T/2)}
{\sqrt{-p_\varepsilon (\theta )}\tilde{q}_\varepsilon (\theta )} 
\end{eqnarray*}
on $\mathcal {D}_\varepsilon \setminus [u_{\varepsilon , -}, u_{\varepsilon , +}]$. 
In Figure \ref {graph_effective}, 
we numerically calculate $u_{\varepsilon , \pm }$ and $\theta ^*_{\varepsilon , \pm }$ 
for $\varepsilon \in (0, 1]$. 
This suggests the modified condition [A4']. 
\end{remark}

\begin{figure}[t]
\begin{center}
\includegraphics[width=10cm]{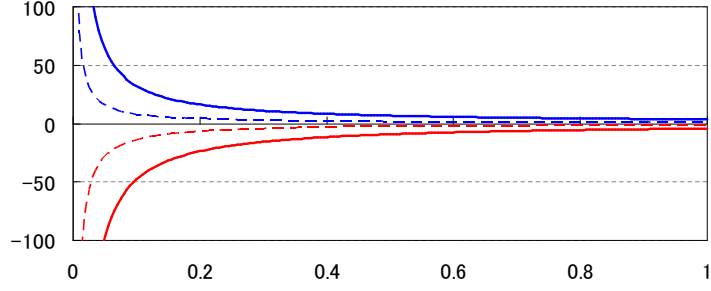}
\caption{Plots of $\theta ^*_{\varepsilon , \pm }$ (solid lines) and 
$u_{\varepsilon , \pm }$ (dashed lines). 
Note that $\theta ^*_{\varepsilon , +} > \theta ^*_{\varepsilon , -}$ and 
$u_{\varepsilon , +} > u_{\varepsilon , -}$. 
The horizontal axis corresponds to $\varepsilon $. }
\label{graph_effective}
\end{center}
\end{figure}

Now we verify the orders of approximate terms 
$\Psi ^\varepsilon _m(\hat{w}_\varepsilon )$ with $m = 0, 1, 2$. 
Figure \ref {graph_LL} 
represents the log-log plot of the approximations for small $\varepsilon $. 
In this figure, we can find the linear relationships between 
$\log |\Psi ^\varepsilon _m(\hat{w}_\varepsilon )|$ and $\log \varepsilon $. 
We estimate their relationship by linear regression and get 
\begin{eqnarray*}
\log |\Psi ^\varepsilon _0(\hat{w}_\varepsilon )| &=& 1.1365\log \varepsilon - 4.5611, \ \ R^2 = 0.9996, \\
\log |\Psi ^\varepsilon _1(\hat{w}_\varepsilon )| &=& 3.3152\log \varepsilon - 7.8203, \ \ R^2 = 0.9999, \\
\log |\Psi ^\varepsilon _2(\hat{w}_\varepsilon )| &=& 4.9068\log \varepsilon - 10.928, \ \ R^2 = 0.9999. 
\end{eqnarray*}
Then we can numerically confirm that 
$\Psi ^\varepsilon _m(\hat{w}_\varepsilon ) = O(\varepsilon ^{2m+1})$ as 
$\varepsilon \rightarrow 0$ for $m = 0, 1, 2$, 
which is consistent with Theorem \ref {th_main} and (\ref {higher_order_estimate}) 
(see also Theorem \ref {th_main2} in Section \ref {sec_error_estimate}). 

\begin{figure}[t]
\begin{center}
\includegraphics[width=10cm]{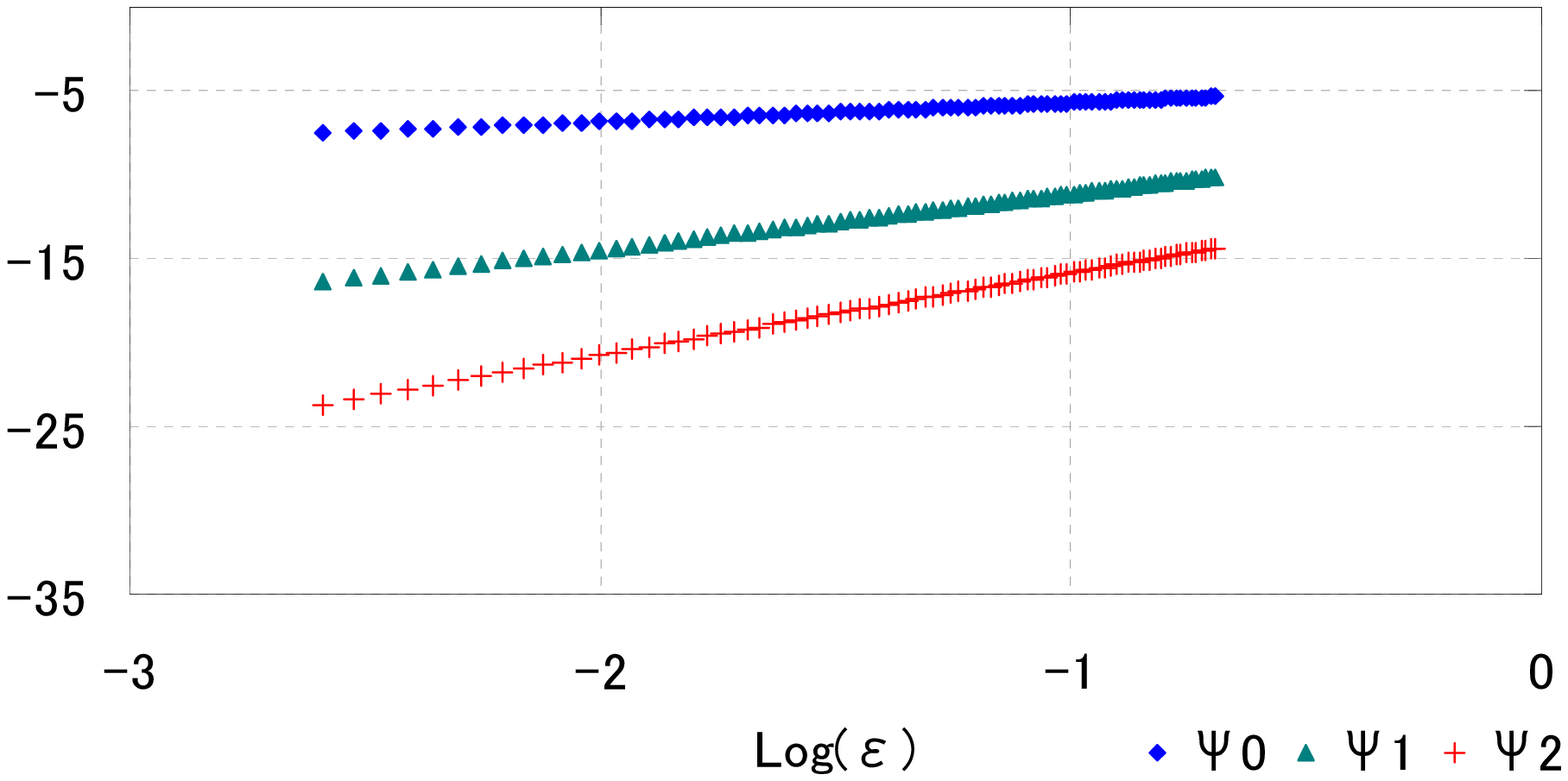}
\caption{Log-log plot of $|\Psi ^\varepsilon _m(\hat{w}_\varepsilon )|$ with $m = 0, 1, 2$ in the Heston SV model. 
The horizontal axis means $\log \varepsilon $. 
The vertical axis means $\log |\Psi ^\varepsilon _m(\hat{w}_\varepsilon )|$. }
\label{graph_LL}
\end{center}
\end{figure}

Next, we calculate the relative errors of the LR formula. 
We let 
\begin{eqnarray*}
\mbox{Normal formula} &=& \bar{\Phi }(\hat{w}_\varepsilon ), \\
0\mbox{th formula} &=& \bar{\Phi }(\hat{w}_\varepsilon ) + \Psi ^\varepsilon _0(\hat{w}_\varepsilon ), \\
1\mbox{st formula} &=& \bar{\Phi }(\hat{w}_\varepsilon ) + \Psi ^\varepsilon _0(\hat{w}_\varepsilon ) + \Psi ^\varepsilon _1(\hat{w}_\varepsilon ),\\
2\mbox{nd formula} &=& \bar{\Phi }(\hat{w}_\varepsilon ) + \Psi ^\varepsilon _0(\hat{w}_\varepsilon ) + \Psi ^\varepsilon _1(\hat{w}_\varepsilon ) + \Psi ^\varepsilon _2(\hat{w}_\varepsilon ). 
\end{eqnarray*}
We define the relative error for approximated value $\hat{P}$ of $P(X^\varepsilon _T > x)$: 
\begin{eqnarray}\label{def_RE}
\mathrm {RE} = \left | \frac{\hat{P}}{P(X^\varepsilon _T > x)}-1 \right |. 
\end{eqnarray}
To find the true value of $P(X^\varepsilon _T > x)$ (`True' in Table \ref {table_RE_Heston}), 
we directly calculate the integral (\ref {calc_tail}) with $c = \hat{\theta }_\varepsilon $. 

\begin{table}[hbtp]
\begin{center}
\scalebox{0.8}[0.8]{
\begin{tabular}{c|ccccc|rrrr} \hline 
 \multirow{2}{*}{$\varepsilon $}	& \multicolumn{5}{c|}{$P(X^\varepsilon _T > x)$}		& \multicolumn{4}{c}{RE} 	\\ \cline {2-10}
 	& True	& Normal	& 0th	& 1st	& 2nd	& \multicolumn{1}{c}{Normal}	& \multicolumn{1}{c}{0th}	& \multicolumn{1}{c}{1st}	& \multicolumn{1}{c}{2nd}	\\ \hline
 0.2	& 0.06622 	& 0.06788 	& 0.06622 	& 0.06622 	& 0.06622 	& 2.51E-02	& 2.84E-05	& 3.12E-07	& 3.18E-09	\\
 0.4	& 0.06521 	& 0.06894 	& 0.06523 	& 0.06521 	& 0.06521 	& 5.71E-02	& 2.88E-04	& 9.57E-06	& 4.04E-07	\\
 0.6	& 0.06385 	& 0.06996 	& 0.06392 	& 0.06385 	& 0.06385 	& 9.56E-02	& 1.11E-03	& 6.76E-05	& 6.28E-06	\\
 0.8	& 0.06219 	& 0.07093 	& 0.06237 	& 0.06217 	& 0.06219 	& 1.41E-01	& 2.82E-03	& 2.60E-04	& 4.11E-05\\
 1	& 0.06029 	& 0.07184 	& 0.06063 	& 0.06025 	& 0.06028 	& 1.92E-01	& 5.69E-03	& 7.22E-04	& 1.40E-04\\\hline 
\end{tabular}
}
\end{center}
\caption{Approximated values of $P(X^\varepsilon _T > x)$ and relative errors with $\varepsilon = 0.2, 0.4, 0.6, 0.8, 1$. 
}
\label{table_RE_Heston}
\end{table}

\begin{figure}[hbtp]
\begin{center}
\includegraphics[width=12cm]{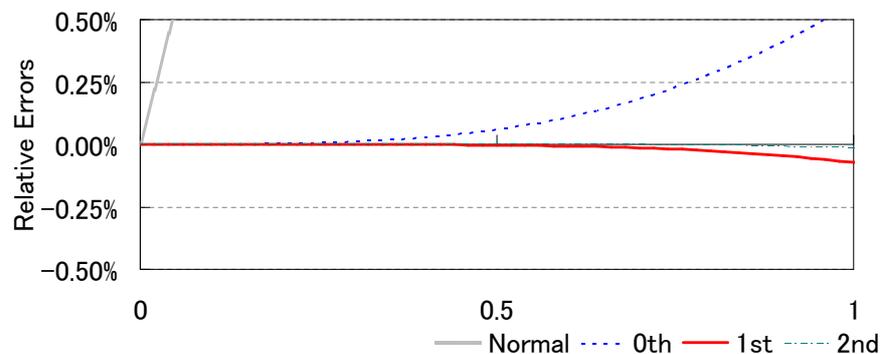}
\caption{Relative errors of $P(X^\varepsilon _T > x)$. 
The horizontal axis means $\varepsilon $. 
}
\label{graph_RE}
\end{center}
\end{figure}

The results are shown in Table \ref {table_RE_Heston} and Figure \ref {graph_RE}. 
We see that the relative errors decrease when $\varepsilon $ becomes small. 
Moreover, we can verify that the higher order LR formula gives a more accurate approximation. 
In particular, the accuracies of the `$1$st' and `$2$nd' formulae are quite high, even when $\varepsilon $ is not small. 

\begin{figure}[t]
\begin{center}
\includegraphics[width=10cm]{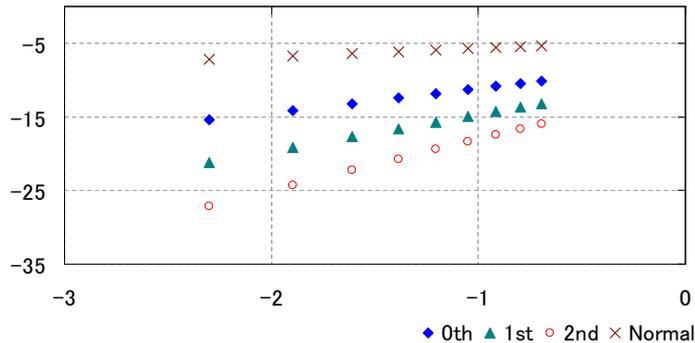}
\caption{Log-log plot of the absolute errors of $P(X^\varepsilon _T > x)$. 
The horizontal axis means $\log \varepsilon $. 
The vertical axis means $\log \mathrm {AE}$. 
}
\label{graph_LL_RE_Heston}
\end{center}
\end{figure}

Figure \ref {graph_LL_RE_Heston} shows the log-log plot of the absolute errors, defined by 
\begin{eqnarray}\label{def_AE}
\mathrm {AE} = \left| \hat{P} - P(X^\varepsilon _T > x) \right|. 
\end{eqnarray}
We see that there are linear relationships between $\log \varepsilon $ and the $\log \mathrm {AE}$ functions: 
by linear regression, we have 
\begin{eqnarray*}
\log \mathrm {AE}_{\mathrm{Normal}} &=& 1.1460\log \varepsilon - 4.5447, \ \ R^2 = 0.9996, \\
\log \mathrm {AE}_{0\mathrm{th}}    &=& 3.2951\log \varepsilon - 7.8692, \ \ R^2 = 0.9999, \\
\log \mathrm {AE}_{1\mathrm{st}}    &=& 4.9353\log \varepsilon - 9.7660, \ \ R^2 = 0.9999, \\
\log \mathrm {AE}_{2\mathrm{nd}}    &=& 6.9894\log \varepsilon - 11.050, \ \ R^2 = 0.9999. 
\end{eqnarray*}
These imply that the error of the $m$th LR formula has order $O(\varepsilon ^{2m + 3})$ as $\varepsilon \rightarrow 0$, 
which is consistent with (\ref {error_estimate}) and Theorem \ref {th_main3} in Section \ref {sec_error_estimate}. 

At the end of this section, we consider the application to option pricing. 
We calculate the European call option price 
\begin{eqnarray}\label{call_eq}
\mathrm {Call}^\varepsilon  = \mathrm {E} [\max \left\{  \exp \left( X^\varepsilon _T\right) - L, 0\right\}  ]
\end{eqnarray}
under the risk-neutral probability measure $P$, where $L > 0$ is the strike price. 

The explicit form of $\mathrm {Call}^\varepsilon $ was obtained by Heston (1993), 
so we can calculate the exact value, up to the truncation error associated with numerical integration. 
Applying the LR formula to (\ref {call_eq}) was proposed by 
Rogers and Zane (1999). 
Here, we briefly review the procedure to do so. 
First, we rewrite (\ref {call_eq}) as 
\begin{eqnarray*}
\mathrm {Call}^\varepsilon  = 
\mathrm {E} [\exp \left( X^\varepsilon _T\right)\ ; \ X^\varepsilon _T > l] - 
LP(X^\varepsilon _T > l), 
\end{eqnarray*}
where $l = \log L$. 
For the second term in the right-hand side of the above equality, 
we can directly apply the LR formula. 
To evaluate the first term, we define a new probability measure $Q$ 
(called the share measure) by the following Radon--Nikodym density 
\begin{eqnarray*}
\frac{dQ}{dP} = \frac{\exp \left(X^\varepsilon _T\right)}{\mathrm {E} [\exp \left( X^\varepsilon _T\right)]} = 
\exp \left( -K_\varepsilon (1)\right) \exp \left( X^\varepsilon _T\right). 
\end{eqnarray*}
From this we obtain 
\begin{eqnarray*}
\mathrm {E} [\exp \left( X^\varepsilon _T\right)\ ; \ X^\varepsilon _T > l] = 
\exp \left( K_\varepsilon (1)\right)Q(X^\varepsilon _T > l). 
\end{eqnarray*}
Now, we can easily find the CGF $\tilde{K}_\varepsilon (\theta )$ of the distribution $Q(X^\varepsilon _T\in \cdot )$: 
\begin{eqnarray*}
\tilde{K}_\varepsilon (\theta ) = K_\varepsilon (\theta + 1) - K_\varepsilon (1). 
\end{eqnarray*}
Obviously, $\tilde{K}_\varepsilon (\theta )$ satisfies our assumptions [A1]--[A5]. 
Therefore, we can apply the LR formula to $Q(X^\varepsilon _T > l)$. 

Now we set the initial price $e^{x_0}$ of the underlying asset as $100$ and 
the strike price $L$ as 105. 
For the model parameters, we set 
$\kappa = 6$, $b = 0.3^2$, $\rho = 0.3$, and $v_0 = 0.2^2$. 
We denote by $\mathrm {Call}^\varepsilon _{\mathrm {Normal}}$, 
$\mathrm {Call}^\varepsilon _{0\mathrm {th}}$, 
$\mathrm {Call}^\varepsilon _{1\mathrm {st}}$ and 
$\mathrm {Call}^\varepsilon _{2\mathrm {nd}}$ the approximations of 
$\mathrm {Call}^\varepsilon $ using the LR formulae `Normal,' `$0$th,' `$1$st' and `$2$nd', 
respectively. 
RE and AE are the same as in (\ref {def_RE}) and (\ref {def_AE}), respectively, 
with tail probabilities as option prices.

\begin{table}[hbtp]
\begin{center}
\scalebox{0.8}[0.8]{
\begin{tabular}{c|ccccc|rrrr} \hline 
 \multirow{2}{*}{$\varepsilon $}	& \multicolumn{5}{c|}{Call Option Price}		& \multicolumn{4}{c}{RE} 	\\ \cline {2-10}
 	& \multicolumn{1}{c}{True}	& \multicolumn{1}{c}{Normal}	& \multicolumn{1}{c}{0th}	& \multicolumn{1}{c}{1st}	& \multicolumn{1}{c|}{2nd}	& \multicolumn{1}{c}{Normal}	& \multicolumn{1}{c}{0th}	& \multicolumn{1}{c}{1st}	& \multicolumn{1}{c}{2nd}	\\ \hline
 0.2	& 9.352 	& 9.367 	& 9.352 	& 9.352 	& 9.352 	& 1.62E-03	& 8.93E-06	& 5.95E-08	& 7.04E-10	\\
 0.4	& 9.358 	& 9.419 	& 9.357 	& 9.358 	& 9.358 	& 6.46E-03	& 1.41E-04	& 3.78E-06	& 1.60E-07	\\
 0.6	& 9.337 	& 9.471 	& 9.330 	& 9.337 	& 9.337 	& 1.43E-02	& 7.00E-04	& 4.29E-05	& 3.43E-06	\\
 0.8	& 9.291 	& 9.523 	& 9.271 	& 9.293 	& 9.291 	& 2.50E-02	& 2.14E-03	& 2.38E-04	& 2.63E-05	\\
 1	& 9.223 	& 9.576 	& 9.177 	& 9.231 	& 9.224 	& 3.82E-02	& 5.01E-03	& 8.79E-04	& 1.16E-04	\\\hline 
\end{tabular}
}
\end{center}
\caption{Approximated values of $\mathrm {Call}^\varepsilon $ and relative errors with $\varepsilon = 0.2, 0.4, 0.6, 0.8, 1$ in the Heston SV model. 
}
\label{table_RE_Heston_Call}
\end{table}

\begin{figure}[t]
\begin{center}
\includegraphics[width=10cm]{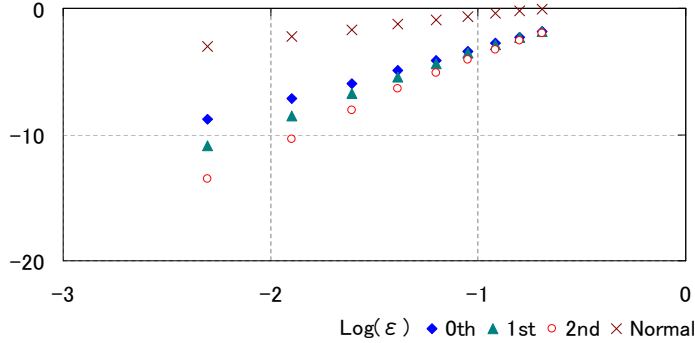}
\caption{Log-log plot of the absolute errors of $\mathrm {Call}^\varepsilon $ in the Heston SV model. 
The horizontal axis means $\log \varepsilon $. 
The vertical axis means $\log \mathrm {AE}$. 
}
\label{graph_LL_AE_Heston_Call}
\end{center}
\end{figure}

Table \ref {table_RE_Heston_Call} and 
Figure \ref {graph_LL_AE_Heston_Call} summerise the results. 
As in the tail probability case, we can see that the LR formulae yield highly accurate approximations.

\subsection{The Wishart SV Model}\label{sec_Wishart}

Next, we introduce the Wishart SV model. 
The Wishart process was first studied by Bru (1991); it was 
first used to describe multivariate stochastic volatility by 
Gouri\'eroux (2006). 
Since then, modelling of multivariate stochastic volatility 
by using the Wishart process has been studied in several papers, such as 
Fonseca, Grasselli, and Tebaldi (2007, 2008), 
Grasselli and Tebaldi (2008), Gouri\'eroux, Jasiak, and Sufana (2009), and 
Benamid, Bensusan, and El Karoui (2010). 

We consider the following SDE: 
\begin{eqnarray*}
&&dY^\varepsilon _t = -\frac{1}{2}\mathrm {tr}[\Sigma ^\varepsilon _t]dt + 
\mathrm {tr}[\sqrt{\Sigma ^\varepsilon _t}(dW_tR' + dB_t\sqrt{I - RR'})] , \\
&&d\Sigma ^\varepsilon _t = 
(\Omega '\Omega + M\Sigma ^\varepsilon _t + \Sigma ^\varepsilon _tM')dt + 
\varepsilon \left\{ \sqrt{\Sigma ^\varepsilon _t}dW_tQ + 
Q'(dW_t)'\sqrt{\Sigma ^\varepsilon _t}\right\} , \\
&&\ Y^\varepsilon _0 = y_0, \ \Sigma ^\varepsilon _0 = \Sigma _0, 
\end{eqnarray*}
where $I$ is the $n$-dimensional unit matrix, 
$R, M, Q\in \Bbb {R}^n\otimes \Bbb {R}^n$, and 
$\varepsilon \geq 0$. 
Here, $\mathrm {tr}[A]$ is the trace of $A$ and $A'$ denotes the transpose matrix of $A$. 
$\Omega \in \Bbb {R}^n\otimes \Bbb {R}^n$ is assumed to satisfy 
\begin{eqnarray*}
\Omega '\Omega = \beta Q'Q
\end{eqnarray*}
for some $\beta \geq (n-1)\varepsilon ^2$. 
$(W_t)_t$ and $(B_t)_t$ are $\Bbb {R}^n\otimes \Bbb {R}^n$-valued processes 
whose components are mutually independent standard Brownian motions. 
The process $(Y^\varepsilon _t)_t$ is regarded as the log-price of a security 
under a risk-neutral probability measure. 
$(\Sigma ^\varepsilon _t)_t$ is an $n$-dimensional matrix-valued process which 
describes multivariate stochastic volatility. 
We verify the validity of the approximation terms of the exact LR expansion for 
$\bar{F}_\varepsilon (x) = P(Y^\varepsilon _T > x)$. 

The explicit form of the CGF of $\mu _\varepsilon = P(Y^\varepsilon _T\in \cdot )$ is 
studied in Bru (1991), Fonseca, Grasselli and Tebaldi (2008), and others. 
To simplify, 
we only treat the case of $n = 2$ and 
restrict the forms of $R, M$ and $Q$ as follows: 
\begin{eqnarray*}
R = 
\left(
\begin{array}{cc}
 	r& 0	\\
 	0& r
\end{array}
\right), \ \ 
M = 
\left(
\begin{array}{cc}
 	-m& 0	\\
 	0& -m
\end{array}
\right), \ \ 
Q = 
\left(
\begin{array}{cc}
 	q& 0	\\
 	0& q
\end{array}
\right), \ \ 
\Sigma _0 = 
\left(
\begin{array}{cc}
 	\sigma ^2_0& 0	\\
 	0& \sigma ^2_0
\end{array}
\right). 
\end{eqnarray*}
We set parameters as 
$r = -0.7$, $q = 0.25$, $m = 1$, $\beta = 3$, $y_0 = 0$, $\sigma _0 = 1$, $T = 1$, and $x = 1$. 
Similar to the case in Section \ref {sec_Heston}, 
we can find linear relationships between 
$\log |\Psi ^\varepsilon _m(\hat{w}_\varepsilon )|$ and $\log \varepsilon $ 
in Figure \ref {graph_LL_Wishart} with $m = 0, 1, 2$. 
Linear regression gives 
\begin{eqnarray*}
\log |\Psi ^\varepsilon _0(\hat{w}_\varepsilon )| &=& 0.9740\log \varepsilon - 4.9496, \ \ R^2 = 0.9999, \\
\log |\Psi ^\varepsilon _1(\hat{w}_\varepsilon )| &=& 2.8128\log \varepsilon - 10.943, \ \ R^2 = 0.9995, \\
\log |\Psi ^\varepsilon _2(\hat{w}_\varepsilon )| &=& 5.2875\log \varepsilon - 14.474, \ \ R^2 = 0.9999. 
\end{eqnarray*}
Thus, for this case also we can numerically confirm that 
$\Psi ^\varepsilon _m(\hat{w}_\varepsilon ) = O(\varepsilon ^{2m+1})$, 
$\varepsilon \rightarrow 0$ for $m = 0, 1, 2$. 

\begin{figure}[t]
\begin{center}
\includegraphics[width=10cm]{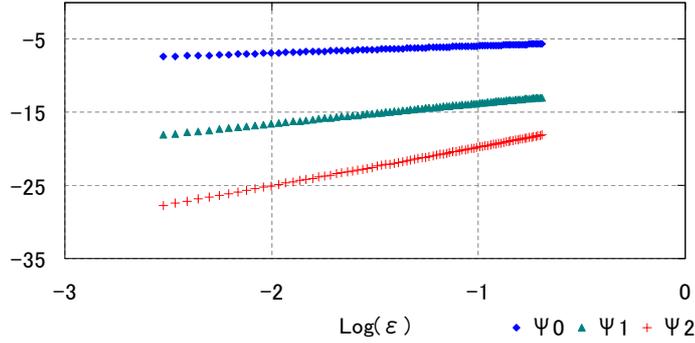}
\caption{Log-log plot of $|\Psi ^\varepsilon _m(\hat{w}_\varepsilon )|$ with $m = 0, 1, 2$ in the Wishart SV model. 
The horizontal axis means $\log \varepsilon $. 
The vertical axis means $\log |\Psi ^\varepsilon _m(\hat{w}_\varepsilon )|$. }
\label{graph_LL_Wishart}
\end{center}
\end{figure}

Now we investigate the relative errors of the LR formula. 
We compare the approximations of $P(Y^\varepsilon _{T} > x)$ 
by the formulae `Normal,' `$0$th,' `$1$st', and `$2$nd', 
defined in the same way as in Section \ref {sec_Heston}, 
with the true value, which is calculated by direct evaluation of the integral in (\ref{calc_tail}).

\begin{table}[hbtp]
\begin{center}
\scalebox{0.8}[0.8]{
\begin{tabular}{c|ccccc|rrrr} \hline 
 \multirow{2}{*}{$\varepsilon $}	& \multicolumn{5}{c|}{$P(Y^\varepsilon _T > x)$}		& \multicolumn{4}{c}{RE} 	\\ \cline {2-10}
 	& True	& Normal	& 0th	& 1st	& 2nd	& \multicolumn{1}{c}{Normal}	& \multicolumn{1}{c}{0th}	& \multicolumn{1}{c}{1st}	& \multicolumn{1}{c}{2nd}	\\ \hline
 0.2	& 0.06610 	& 0.06462 	& 0.06610 	& 0.06610 	& 0.06610 	& 2.24E-02	& 2.97E-06	& 4.62E-09	& 1.73E-11	\\
 0.4	& 0.06624 	& 0.06333 	& 0.06623 	& 0.06624 	& 0.06624 	& 4.38E-02	& 2.00E-05	& 1.88E-07	& 2.53E-09	\\
 0.6	& 0.06622 	& 0.06198 	& 0.06622 	& 0.06622 	& 0.06622 	& 6.40E-02	& 5.25E-05	& 1.77E-06	& 4.76E-08	\\
 0.8	& 0.06604 	& 0.06056 	& 0.06603 	& 0.06604 	& 0.06604 	& 8.30E-02	& 8.31E-05	& 7.92E-06	& -6.20E-07	\\
 1	& 0.06568 	& 0.05908 	& 0.06567 	& 0.06568 	& 0.06568 	& 1.00E-01	& 9.12E-05	& 4.59E-06	& -2.52E-05\\\hline 
\end{tabular}
}
\end{center}
\caption{Approximated values of $P(Y^\varepsilon _T > x)$ and relative errors for $\varepsilon = 0.2, 0.4, 0.6, 0.8, 1$. 
}
\label{table_RE_Wishart}
\end{table}

\begin{figure}[t]
\begin{center}
\includegraphics[width=10cm]{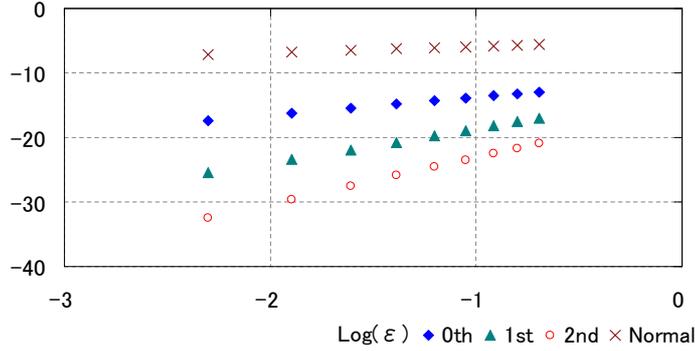}
\caption{Log-log plot of the absolute errors of $P(Y^\varepsilon _T > x)$. 
The horizontal axis means $\log \varepsilon $. 
The vertical axis means $\log \mathrm {AE}$. 
}
\label{graph_LL_AE_Wishart}
\end{center}
\end{figure}

Similar to the case in Section \ref {sec_Heston}, 
we show the relative errors and the log-log plot of absolute errors of the formulae in 
Table \ref {table_RE_Wishart} and Figure \ref {graph_LL_AE_Wishart}. 
We can also confirm that the LR formulae are highly accurate. 
Using the data shown in Figure \ref {graph_LL_AE_Wishart}, we get the linear regression results 
\begin{eqnarray*}
\log \mathrm {AE}_{\mathrm{Normal}} &=& 0.9732\log \varepsilon - 4.9507, \ \ R^2 = 0.9999, \\
\log \mathrm {AE}_{0\mathrm{th}}    &=& 2.7930\log \varepsilon - 10.979, \ \ R^2 = 0.9994, \\
\log \mathrm {AE}_{1\mathrm{st}}    &=& 5.3063\log \varepsilon - 13.339, \ \ R^2 = 0.9999, \\
\log \mathrm {AE}_{2\mathrm{nd}}    &=& 7.1747\log \varepsilon - 15.937, \ \ R^2 = 0.9999,  
\end{eqnarray*}
which suggest (\ref {error_estimate}). 

At the end of this section, we confirm the validity for application in option pricing. 
Similarly to (\ref {call_eq}), we consider the European call option 
\begin{eqnarray*}
\mathrm {Call}^\varepsilon  = \mathrm {E} [\max \left\{  \exp \left( Y^\varepsilon _T\right) - L, 0\right\}  ]
\end{eqnarray*}
with the strike price $L > 0$. 
To find the true value of the option price, we apply a closed-form formula 
proposed in Benabid, Bensusan, and El Karoui (2010). 
We set the initial price of the underlying asset as $e^{y0} = 100$ and $L = 105$. 
For the initial volatility, we put $\sigma _0 = 0.25$. 
Other parameters are the same as in the previous case. 

\begin{table}[hbtp]
\begin{center}
\scalebox{0.8}[0.8]{
\begin{tabular}{c|ccccc|rrrr} \hline 
 \multirow{2}{*}{$\varepsilon $}	& \multicolumn{5}{c|}{Call Option Price}		& \multicolumn{4}{c}{RE} 	\\ \cline {2-10}
 	& \multicolumn{1}{c}{True}	& \multicolumn{1}{c}{Normal}	& \multicolumn{1}{c}{0th}	& \multicolumn{1}{c}{1st}	& \multicolumn{1}{c|}{2nd}	& \multicolumn{1}{c}{Normal}	& \multicolumn{1}{c}{0th}	& \multicolumn{1}{c}{1st}	& \multicolumn{1}{c}{2nd}	\\ \hline
 0.2	& 10.90 	& 10.91 	& 10.90 	& 10.90 	& 10.90 	& 1.13E-03	& 1.50E-06	& 8.61E-09	& 3.15E-11	\\
 0.4	& 10.76 	& 10.80 	& 10.76 	& 10.76 	& 10.76 	& 4.58E-03	& 2.10E-05	& 4.56E-05	& 4.62E-05	\\
 0.6	& 10.60 	& 10.70 	& 10.59 	& 10.59 	& 10.59 	& 9.88E-03	& 4.37E-04	& 3.08E-04	& 3.02E-04	\\
 0.8	& 10.46 	& 10.60 	& 10.40 	& 10.41 	& 10.41 	& 1.27E-02	& 5.71E-03	& 5.29E-03	& 5.25E-03	\\
 1	& 10.15 	& 10.49 	& 10.20 	& 10.21 	& 10.21 	& 3.37E-02	& 4.33E-03	& 5.41E-03	& 5.55E-03\\\hline 
\end{tabular}
}
\end{center}
\caption{Approximations of $\mathrm {Call}^\varepsilon $ and relative errors with $\varepsilon = 0.2, 0.4, 0.6, 0.8, 1$ in the Wishart SV model. 
}
\label{table_RE_Wishart_Call}
\end{table}

\begin{figure}[t]
\begin{center}
\includegraphics[width=10cm]{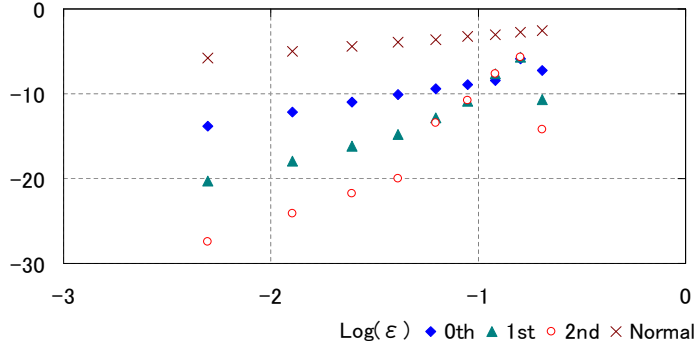}
\caption{Log-log plot of the absolute errors of $\mathrm {Call}^\varepsilon $ in the Wishart SV model. 
The horizontal axis means $\log \varepsilon $. 
The vertical axis means $\log \mathrm {AE}$. 
}
\label{graph_LL_AE_Wishart_Call}
\end{center}
\end{figure}

The results are shown in Table \ref {table_RE_Wishart_Call} and Figure \ref {graph_LL_AE_Wishart_Call}. 
Although the linear relationships are not as clear as in Figure \ref {graph_LL_AE_Heston_Call}, 
we can see that the LR formulae are highly accurate in each case.

\section{Proofs}\label{sec_proofs}

\subsection{Proof of Theorem \ref {th_exact_LR}}
\label{sec_proof_LR}

In this subsection, we justify the formal calculations shown in Section \ref {sec_LR}. 
For ease of readability, we omit $\varepsilon $ from the notation used in this section. 

\begin{proposition}\label{prop_inversion} \ 
Assume $\mathrm {[A1]}$--$\mathrm {[A2]}$ hold. 
Then 
\begin{eqnarray*}
\bar{F}(x) \ = \ 
\frac{1}{2\pi i}\int ^{c + i\infty }_{c - i\infty }\exp (K (\theta ) - x\theta  )\frac{d\theta }{\theta } 
\end{eqnarray*}
for $c \in \mathcal {O}\setminus \{0\}$. 
\end{proposition}

\begin{proof} 
Without loss of generality, we may assume $c > 0$ and $x \geq 0$. 
By [A2] and Theorem 3.3.5 in Durrett (2010), 
the density function $f$ of $\mu $ exists and is bounded and continuous. 
Moreover, 
\begin{eqnarray*}
f(y) = \frac{1}{2\pi }\int _{\Bbb {R}}e^{-i\xi y}\varphi (\xi )d\xi  
\end{eqnarray*}
holds, where $\varphi (\xi ) = \exp (K(i\xi ))$ is the characteristic function of $\mu $. 
Then we have for each $R > x$ that 
\begin{eqnarray*}
\mu ((x, R]) = \frac{1}{2\pi }\int ^R_x\int _{\Bbb {R}}e^{-i\xi y}\varphi (\xi )d\xi dy = 
\frac{1}{2\pi i}\int ^{i\infty }_{-i\infty }F(z)dz
\end{eqnarray*}
by Fubini's theorem, where 
\begin{eqnarray*}
F(z) = 
\int _{\Bbb {R}}\int ^R_xe^{(s - y)z}f(s)dyds. 
\end{eqnarray*}
Now, consider the four lines $\Gamma _1, \ldots , \Gamma _4\subset \Bbb {C}$, defined as 
\begin{eqnarray*}
&&\Gamma _1 = \{ it\ ; \ t\in [-l, l] \}, \ \ 
\Gamma _2 = \{ t + il\ ; \ t\in [0, c] \}, \\
&&\Gamma _3 = \{ t - il\ ; \ t\in [0, c] \}, \ \ 
\Gamma _4 = \{ c + it\ ; \ t\in [-l, l] \} 
\end{eqnarray*}
for a given $l > 0$. 
By Cauchy's integral theorem, we have 
\begin{eqnarray}\label{int_Cauchy}
\int _{\Gamma _1\cup \Gamma _2\cup \Gamma _3\cup \Gamma _4}F(z)dz = 0. 
\end{eqnarray}
Here, we observe that
\begin{eqnarray*}
\int _{\Gamma _2\cup \Gamma _3}F(z)dz 
&=& 
2i\int ^c_0\int _{\Bbb {R}}\int ^R_xe^{(s-y)t}f(s)\sin l(s-y)dydsdt\\
&=& 
2i\int ^c_0\int _{\Bbb {R}}\frac{f(s)}{t^2 + l^2}
\Big[ e^{(s-x)t}(-l\cos l(s-x) + t\sin l(s-x))\\&&\hspace{30mm} + 
e^{(s-R)t}(l\cos l(s-R) - t\sin l(s-R))\Big] dsdt
\end{eqnarray*}
to conclude 
\begin{eqnarray*}
\left | \int _{\Gamma _2\cup \Gamma _3}F(z)dz \right | \ \leq \ 
\frac{4(l+c)c}{l^2}\int _{\Bbb {R}} e^{cs}f(s)ds = 
\frac{4(l+c)c}{l^2}e^{K(c)}. 
\end{eqnarray*}
Since $c\in \mathcal {O}$, the integral on the right-hand side is finite. 
Thus, the left-hand side must converge to zero as $l\rightarrow \infty $. 
Combining this result with (\ref {int_Cauchy}), we obtain that
\begin{eqnarray}\nonumber 
&&\mu ((x, R]) \ =\  \frac{1}{2\pi i}\lim _{l\rightarrow \infty }\int _{\Gamma _1}F(z)dz \ =\  
\frac{1}{2\pi i}\lim _{l\rightarrow \infty }\int _{\Gamma _4}F(z)dz\\&&= 
\frac{1}{2\pi i}\int ^{c+i\infty }_{c-i\infty }
\int _{\Bbb {R}}\int ^R_xe^{(s - y)z}f(s)dydsdz\ = \ 
\frac{1}{2\pi i}\int ^{c+i\infty }_{c-i\infty }
\int ^R_xe^{K(z)-yz}dydz. \ \ \ \ \ 
\label{temp_integral}
\end{eqnarray}
Since 
\begin{eqnarray*}
\int ^{c+i\infty }_{c-i\infty }\int ^\infty _0|e^{K(z) - yz}|dy|dz| \leq 
\frac{1}{c}e^{K(c)} < \infty , 
\end{eqnarray*}
we can take the limit $R\rightarrow \infty $ on the right-hand side of (\ref {temp_integral}); we conclude that
\begin{eqnarray*}
\bar{F}(x) = \frac{1}{2\pi i}\int ^{c+i\infty }_{c-i\infty }
\int ^\infty _xe^{K(z)-yz}dydz = 
\frac{1}{2\pi i}\int ^{c+i\infty }_{c-i\infty }e^{K(z) - xz}\frac{dz}{z}, 
\end{eqnarray*}
which is the assertion of Proposition \ref {prop_inversion}. 
\end{proof}

Now, we present the rigorous definition of the change of variables (\ref {change_variable}). 
For each $\theta \in \mathcal {D}$, we can define $w = w(\theta )\in \Bbb {R}$ by 
\begin{eqnarray}\label{def_w_real}
w(\theta ) = 
\hat{w} + \mathrm {sgn}(\theta - \hat{\theta })
\sqrt{2\left \{ \left( K(\theta ) - x\theta \right) - 
\left( K(\hat{\theta }) - x\hat{\theta }\right) \right \}}. 
\end{eqnarray}
Obviously, $w(\theta )$ is analytic on $\mathcal {O}\setminus \{\hat{\theta }\}$. 
Moreover, by straightforward calculation we observe 
\begin{eqnarray}\label{w_diff}
\frac{dw}{d\theta } = \mathrm {sgn}(\theta - \hat{\theta })\cdot  \frac{K'(\theta ) - x}{\sqrt{\hat{w}^2 + 2(K(\theta ) - x\theta )}} = 
\frac{K'(\theta ) - x}{w(\theta ) - \hat{w}}. 
\end{eqnarray}
Here we see that $w(\theta )$ is also analytic at $\hat{\theta }$. 
Indeed, similar to (\ref {calc_K_diff}), we have 
\begin{eqnarray*}
w(\theta ) = \hat{w} + \mathrm {sgn}(\theta - \hat{\theta })
\sqrt{k(\theta )}|\theta - \hat{\theta }| = 
\hat{w} + 
\sqrt{k(\theta )}(\theta - \hat{\theta }), 
\end{eqnarray*}
where 
\begin{eqnarray*}
k(\theta ) = \int ^1_0K''(\hat{\theta } + u(\theta - \hat{\theta }))du. 
\end{eqnarray*}
By [A2], $k(\theta )$ is positive, and thus $\sqrt{k(\theta )}$ is real analytic. 
As a consequence, the function $w(\theta )$ is real analytic on $\mathcal {O}$. 
Now we can take the limit $\theta \rightarrow \hat{\theta }$ in (\ref {theta_diff}) to obtain 
\begin{eqnarray*}
w'(\hat{\theta }) = \lim _{\theta \rightarrow \hat{\theta }}\frac{K'(\theta ) - x}{w - \hat{w}} = 
\lim _{\theta \rightarrow \hat{\theta }}\frac{K''(\theta )}{w'(\theta )} 
\end{eqnarray*}
by using l'H\^opital's rule. 
This implies that $(w'(\hat{\theta }))^2 = K''(\hat{\theta }) \neq 0$. 
Therefore, we deduce that 
there exist a neighbourhood $U\subset \Bbb {C}$ of $w(\hat{\theta }) = \hat{w}$ and 
a holomorphic function $\theta (w)$ on $U$ such that $\theta (w(z)) = z$ for $z\in U$. 

Here we remark that 
\begin{lemma}\label{lem_K_diff_real} \ 
$\theta \notin \mathcal{D}$ implies $K'(\theta )$ does not lie on $\Bbb {R}$. 
\end{lemma}

\begin{proof} 
Let $y\in \Bbb {R}$. By [A2], we have $K''(y) \neq 0$. 
Thus, we can find a neighbourhood $U$ of $K'(y)$ and 
an analytic inverse function function $(K')^{-1}$ of $K'$ defined on $U$. 
On the other hand, [A2] implies that $(K')^{-1}|_{U\cap \Bbb {R}}$ is an analytic $\mathcal {D}$-valued function, 
hence $(K')^{-1}(y)\in \mathcal {D}$. 
\end{proof}

Lemma \ref {lem_K_diff_real} immediately implies 

\begin{corollary}\label{cor_SP_unique} \ Let $z\in \mathcal {D}\times i\Bbb {R}$. 
If $z \neq \hat{\theta }$, then $K'(z) \neq x$; hence, $w'(\theta ) \neq 0$. 
\end{corollary}

Now, we consider an analytic continuation of $\theta (w)$. 
Until the end of this section, 
we will assume [A1]--[A2] and [B1]--[B3] hold. 
By [B2], (\ref {w_diff}) and Corollary \ref {cor_SP_unique}, 
we can define the analytic function $\theta (w)$ on an open set $\hat{U}$ 
which contains a convex set that includes the line $\{ \hat{w} \} \times i\Bbb {R}$ and 
the curve $\{ w(\hat \theta  + it) \ ; \ t\in \Bbb {R} \}$. 
Note that (\ref {w_diff}) immediately implies 
\begin{eqnarray}\label{theta_diff}
\theta '(w) = \frac{w - \hat{w}}{K'(\theta (w)) - x} 
\end{eqnarray}
for each $w\neq \hat{w}$. 

By definition, the relation (\ref {change_variable}) holds everywhere on $\hat{U}$. 
Therefore, if we define the curves $\eta $ and $\gamma $ as 
\begin{eqnarray*}
\eta  = \{ \hat{\theta } + it\ ; \ t \in \Bbb {R}\} , \ \ 
\gamma  = \{ w(\theta )\ ; \ \theta \in \eta  \} , 
\end{eqnarray*}
then $\theta (w)$ can be also defined and is analytic on $\gamma $. 
Then, we can apply the change of variables to obtain 
\begin{eqnarray*}
\bar{F}(x) \ = \ 
\frac{1}{2\pi i}\int _{\eta }\exp (K (\theta ) - x\theta  )\frac{d\theta }{\theta } \ = \ 
\int _{\gamma }\exp \left (\frac{1}{2}w^2 - \hat{w}w\right )\frac{\theta '(w)}{\theta (w)}dw. 
\end{eqnarray*}

In Section \ref {sec_LR}, we need the condition $\hat{\theta } \neq 0$. 
In this section we only consider the case where $\hat{\theta } > 0$; 
the arguments are analogous for the case where $\hat{\theta } < 0$. 
In any case, we have $\hat{\theta } \neq 0$ and thus $\eta $ does not pass $0$. 
Here, we see that $\hat{w} > 0$. 
Indeed, if $\hat{w} = 0$, then the inequality in (\ref {calc_K_diff}) must be changed to equality. 
However, the assumptions $\hat{\theta } > 0$ and [A2] imply that the left-hand side of (\ref {calc_K_diff}) is positive. 
This is a contradiction. 
Moreover, by its definition, $\hat{w}$ must be nonnegative. 
These arguments imply that $\gamma $ does not exceed $0$. 

\begin{proposition}\label{prop_Cauchy_w} \ 
\begin{eqnarray*}
\int _{\gamma }\exp \left (\frac{1}{2}w^2 - \hat{w}w\right )\frac{\theta '(w)}{\theta (w)}dw &=& 
\int ^{\hat{w} + i\infty }_{\hat{w} - i\infty }\exp \left (\frac{1}{2}w^2 - \hat{w}w\right )\frac{\theta '(w)}{\theta (w)}dw. 
\end{eqnarray*}
\end{proposition}

To prove this proposition, we prepare a lemma. 

\begin{lemma}\label{lem_abs} \ 
$|\theta (w) - \hat{\theta }| \geq |w - \hat{w}|/\sqrt{C}$. 
\end{lemma}

\begin{proof} 
By [B1] and Taylor's theorem, we have 
\begin{eqnarray*}
|w - \hat{w}|^2 = 2| K(\theta (w)) - x\theta (w) - (K(\hat{\theta }) - x\hat{\theta }) | 
\leq C|\theta (w) - \hat{\theta }|^2, 
\end{eqnarray*}
which implies the asserted statement. 
\end{proof}

\begin{proof}[Proof of Proposition \ref {prop_Cauchy_w}] 
By Cauchy's integral theorem, it suffices to show that 
\begin{eqnarray}\label{temp_Lcl0}
\lim _{l\rightarrow \pm \infty }\sup _{c\in \mathcal {O} \cap M}
\left| \int _{L^c_l}\exp \left (\frac{1}{2}w^2 - \hat{w}w\right )\frac{\theta '(w)}{\theta (w)}dw\right| = 0 
\end{eqnarray}
for each compact set $M$ in $(0, \infty )$, 
where $L^c_l = \{ \hat{w} + t(c-\hat{w}) + il\ ;\ t \in [0, 1] \}$ and $l\in \Bbb {R}$. 
By (\ref {theta_diff}), we get 
\begin{eqnarray}\nonumber 
&&\left| \int _{L^c_l}\exp \left (\frac{1}{2}w^2 - \hat{w}w\right )\frac{\theta '(w)}{\theta (w)}dw\right| \\\nonumber 
&\leq & 
|c-\hat{w}|\exp \left( -\frac{\hat{w}^2 + l^2}{2}\right) 
\int ^1_0\exp \left ( \frac{1}{2}t^2(c-\hat{w})^2 \right) 
\left| \frac{t(c - \hat{w}) + il}{(K'(\theta ) - x)\theta }\right|dt \\
&\leq & 
\exp \left( \frac{c^2 - l^2}{2}\right) 
\frac{(|c| + |\hat{w}|)(|c| + |\hat{w}| + |l|)}{\inf _{w\in L^c_l}|(K'(\theta (w)) - x)\theta (w)|}. 
\label{temp_Lcl1}
\end{eqnarray}
By [B1] and Lemma \ref {lem_abs}, we observe
\begin{eqnarray*}
\inf _{w\in L^c_l}|(K'(\theta (w)) - x)\theta (w)| \geq  
\delta \inf _{w\in L^c_l}|\theta (w) - \hat{\theta }|\inf _{w\in L^c_l}|\theta | \geq 
\delta \cdot \frac{|l|}{\sqrt{C}}\cdot \left( \frac{|l|}{\sqrt{C}} - |\hat{\theta }|\right) 
\end{eqnarray*}
for sufficiently large magnitudes of $l$. 
Hence, we obtain (\ref {temp_Lcl0}) from (\ref {temp_Lcl1}). 
\end{proof}

\begin{proof}[Proof of Theorem \ref {th_exact_LR}] 
From Propositions \ref {prop_inversion} and \ref {prop_Cauchy_w}, we get (\ref {Cauchy_bar_F}). 
Now we verify the holomorphicity of $\psi $ on $\{\hat{w}\} \times i\Bbb {R}$. 
We define 
\begin{eqnarray}\label{def_g}
h(w) = \log \theta (w) - \log w = \log g(w), \ \ g(w) = \frac{\theta (w)}{w}
\end{eqnarray}
when $\theta (w)$ is defined and let $w\neq 0$, 
where $\log z$ is the principal value of the logarithm of $z$. 
Since $\theta (w)$ is analytic on the line $\{\hat{w}\} \times i\Bbb {R}$, $h$ is also analytic. 
We can easily see that $h'(w) = \psi (w)$. 
This implies that $\psi (w)$ is also analytic; 
this permits the following Taylor series expansion: 
\begin{eqnarray}\label{psi_expansion}
\psi (w) = \sum ^\infty _{n = 0}\frac{\psi ^{(n)}(\hat{w})}{n!}(w - \hat{w})^n 
\end{eqnarray}
for $w\in \{\hat{w}\}\times i\Bbb {R}$. 

To complete the proof of Theorem \ref {th_exact_LR}, 
it suffices to check the calculations in (\ref {integral_sum}). 
Using (\ref {psi_expansion}) and the relation 
\begin{eqnarray}\label{int_normal}
\int ^\infty _{-\infty }e^{-y^2/2}y^ndy = \sqrt{2\pi }(n - 1)!! \ (\mbox{$n$ is even}), \ 0 \ (\mbox{$n$ is odd}), 
\end{eqnarray}
we have 
\begin{eqnarray*}
&&\sum ^\infty _{n = 0}\frac{1}{n!}\int ^\infty _{-\infty }e^{-y^2/2}|\psi ^{(n)}(\hat{w})|\cdot |y|^ndy \\
&\leq & 
\sqrt{2\pi }|\psi (\hat{w})| + 
\sum ^\infty _{m = 1}\frac{1}{(2m)!}\{ |\psi ^{(2m)}(\hat{w})| + |\psi ^{(2m-1)}(\hat{w})| \} 
\int ^\infty _{-\infty }e^{-y^2/2}(y^{2m} + 1)dy\\ 
&\leq & 
\sqrt{2\pi }\left\{ |\psi (\hat{w})| + 
2\sum ^\infty _{m = 1}\frac{|\psi ^{(2m)}(\hat{w})| + |\psi ^{(2m-1)}(\hat{w})|}{(2m)!!} \right\}. 
\end{eqnarray*}
By [B3], the right-hand side of the above inequality is finite. 
Thus, we can apply Fubini's theorem and we can interchange the sum and the integral in (\ref {integral_sum}). 
That is, 
\begin{eqnarray*}
\int ^\infty _{-\infty }e^{-y^2/2}\sum ^\infty _{n = 0}\frac{\psi ^{(n)}(\hat{w})}{n!}(iy)^ndy  = 
\sum ^\infty _{n = 0}\frac{\psi ^{(n)}(\hat{w})}{n!}i^n
\int ^\infty _{-\infty }e^{-y^2/2}y^ndy. 
\end{eqnarray*}
We finish the proof of Theorem \ref {th_exact_LR} by using (\ref {int_normal}) again. 
\end{proof}

\subsection{Proof of Theorem \ref {th_main}} \label{sec_proof_estimate}

For simplicity, we only consider the case $\hat{\theta }_0 > 0$. 
First, we introduce the following lemma. 

\begin{lemma}\label{lem_SP_conti} \ 
$\hat{\theta }_\varepsilon \longrightarrow \hat{\theta }_0$, \ 
$\hat{w}_\varepsilon \longrightarrow \hat{w}_0$ as $\varepsilon \rightarrow 0$. 
\end{lemma}

\begin{proof} 
First, we check that $(\hat{\theta }_\varepsilon )_\varepsilon $ is bounded. 
By (\ref {def_SP}), we have 
\begin{eqnarray*}
\hat{\theta }_\varepsilon 
= (K'_\varepsilon )^{-1}(x) - (K'_\varepsilon )^{-1}(m_\varepsilon ) = 
\int ^1_0\frac{du}{K''_\varepsilon ((K'_\varepsilon )^{-1}(m_\varepsilon  + u(x - m_\varepsilon )))}(x-m_\varepsilon ), 
\end{eqnarray*}
where $m_\varepsilon = K'_\varepsilon (0)$. 
By [A5], we see that $(m_\varepsilon )_\varepsilon $ is bounded. 
Thus, from [A3], we get 
\begin{eqnarray*}
|\hat{\theta }_\varepsilon | \leq \frac{1}{\delta _0}(|x| + \max _\varepsilon |m_\varepsilon |) < \infty . 
\end{eqnarray*}
Second, we observe that 
\begin{eqnarray*}
x - K'_\varepsilon (\hat{\theta }_0) = K''_\varepsilon (\hat{\theta }_0)(\hat{\theta }_\varepsilon - \hat{\theta }_0) + 
\frac{1}{2}\int ^1_0(1-u)^2K'''_\varepsilon (\hat{\theta }_0 + u(\hat{\theta }_\varepsilon - \hat{\theta }_0))du(\hat{\theta }_\varepsilon - \hat{\theta }_0)^2
\end{eqnarray*}
to arrive at 
\begin{eqnarray*}
|\hat{\theta }_\varepsilon - \hat{\theta }_0| \leq 
\frac{1}{\delta _0}\left\{ |x - K'_\varepsilon (\hat{\theta }_0)| + 
\frac{1}{2}\sup _{y\in C}|K'''_\varepsilon (y)|\cdot \sup _\varepsilon |\hat{\theta }_\varepsilon - \hat{\theta }_0|^2\right\} 
\end{eqnarray*}
for some compact set $C\subset \Bbb {R}$. 
Letting $\varepsilon \rightarrow 0$, we get the former assertion. 
The latter assertion follows immediately. 
\end{proof}

The above lemma implies the following corollary.

\begin{corollary}\label{cor_not_zero} \ 
There is a $\delta _1 > 0$ such that 
$\hat{\theta }_\varepsilon , \ \hat{w}_\varepsilon  > 0$ for $\varepsilon \in [0, \delta _1)$. 
\end{corollary}

\begin{proof} 
Since 
$\hat{\theta }_\varepsilon \longrightarrow \hat{\theta }_0 > 0$, we can find some $\delta _1 > 0$ such that 
$\hat{\theta }_\varepsilon > \hat{\theta }_0 / 2 > 0$ holds for $\varepsilon < \delta _1$. 
The relation $\hat{w}_\varepsilon > 0$ is obtained 
in the same way by using $\hat{w}_0 = \sqrt{K''_0(0)}\hat{\theta }_0 > 0$. 
\end{proof}

By the above corollary, 
we may assume that $\hat{\theta }_\varepsilon $ and $\hat{w}_\varepsilon $ are strictly positive. 

\begin{proposition} \ \label{diff_w_theta}
$\hat{w}_\varepsilon - \sqrt{K''_\varepsilon (\hat{\theta }_\varepsilon )}\hat{\theta }_\varepsilon = 
O(\varepsilon )$ as $\varepsilon \rightarrow 0$. 
\end{proposition}

\begin{proof} 
Since $(\hat{\theta }_\varepsilon )_\varepsilon $ and 
$(\hat{w}_\varepsilon )_\varepsilon $ are bounded and away from zero, 
it suffices to show that 
$\hat{w}^2_\varepsilon - K''_\varepsilon (\hat{\theta }_\varepsilon )\hat{\theta }^2_\varepsilon = O(\varepsilon )$ 
as $\varepsilon \rightarrow 0$. 
From the definition of $\hat{w}_\varepsilon $, we have 
\begin{eqnarray*}
\hat{w}_\varepsilon ^2 = 2(-K_\varepsilon (\hat{\theta }_\varepsilon ) + K'_\varepsilon (\hat{\theta }_\varepsilon )\hat{\theta }_\varepsilon ). 
\end{eqnarray*}
Using $K_\varepsilon (0) = 0$ and Taylor's theorem, we get 
\begin{eqnarray*}
\hat{w}^2_\varepsilon - K''_\varepsilon (\hat{\theta }_\varepsilon )\hat{\theta }^2_\varepsilon = 
\hat{\theta }_\varepsilon ^2\int ^1_0K'''_\varepsilon (-u\hat{\theta }_\varepsilon )(1-u)^2du. 
\end{eqnarray*}
Therefore, 
\begin{eqnarray*}
|\hat{w}^2_\varepsilon - K''_\varepsilon (\hat{\theta }_\varepsilon )\hat{\theta }^2_\varepsilon |\leq 
\sup _{|y|\leq \hat{\theta }_\varepsilon }y^2|K'''_\varepsilon (y)| = O(\varepsilon ) \ \ \mbox {as} \ \ \varepsilon \rightarrow 0, 
\end{eqnarray*}
from which our assertion follows. 
\end{proof}

We write 
\begin{eqnarray*}
\hat{\theta }'_\varepsilon  = \frac{d\theta }{dw}(\hat{w}_\varepsilon ) = 
\lim _{w\rightarrow \hat{w}_\varepsilon }\frac{d\theta }{dw}(\hat{w}_\varepsilon ) . 
\end{eqnarray*}
Note that $\hat{\theta }'_\varepsilon $ exists, because 
$\theta (w)$ is analytic at $\hat{w}_\varepsilon $. 
Similarly, we can define 
\begin{eqnarray*}
\hat{\theta }^{(n)}_\varepsilon  = \frac{d^n\theta }{dw^n}(\hat{w}_\varepsilon ) = 
\lim _{w\rightarrow \hat{w}_\varepsilon }\frac{d^n\theta }{dw^n}(\hat{w}_\varepsilon ) 
\end{eqnarray*}
for each $n$. 
The next proposition is frequently used in the calculations shown later. 

\begin{proposition}\label{first_diff_theta} \ 
$\hat{\theta }'_\varepsilon  = 1/\sqrt{K''(\hat{\theta }_\varepsilon )}$. 
\end{proposition}

\begin{proof} 
Since both the numerator and the denominator in the right-hand side of (\ref {theta_diff}) 
converge to zero with $w\rightarrow \hat{w}_\varepsilon $, 
we can apply l'H\^opital's rule to obtain 
\begin{eqnarray*}
\hat{\theta }'_\varepsilon = 
\lim _{w\rightarrow \hat{w}_\varepsilon }\frac{w - \hat{w}_\varepsilon }{K'(\theta (w)) - x} = 
\lim _{w\rightarrow \hat{w}_\varepsilon }\frac{1}{K''(\theta (w))\theta '(w)} = 
\frac{1}{K''(\hat{\theta }_\varepsilon )\hat{\theta }'_\varepsilon }. 
\end{eqnarray*}
Solving this equation for $\hat{\theta }'_\varepsilon $, we obtain the desired assertion. 
\end{proof}

Recall that the function $g(w)$ defined in (\ref {def_g}) is analytic on 
$\hat{\mathcal {O}}_{\varepsilon  ,+} := 
\{ w(\theta )\ ; \ \theta \in \mathcal {O}_\varepsilon \cap (0, \infty ) \}$. 
The following lemma is straightforward by using mathematical induction.

\begin{lemma} \ \label{lem_g_diff}For each $n = 0, 1, 2, \ldots $ and 
$w\in \hat{\mathcal {O}}_{\varepsilon  ,+}$, 
\begin{eqnarray*}
g^{(n)}(w) = \frac{\theta ^{(n)}(w) - ng^{(n-1)}(w)}{w}. 
\end{eqnarray*}
\end{lemma}

Note that $g(\hat{w}_\varepsilon ) = \hat{\theta }_\varepsilon / \hat{w}_\varepsilon > 0$. 
Therefore, we can define $h(w) = \log g(w)$ on a neighbourhood of $\hat{w}_\varepsilon $. 
Obviously, we have $\psi (w) = h'(w)$. Hence, 
\begin{eqnarray}\label{temp2}
\psi ^{(2m)}(\hat{w}_\varepsilon ) = h^{(2m+1)}(\hat{w}_\varepsilon ). 
\end{eqnarray}
Different but nevertheless 
straightforward calculations give 
\begin{eqnarray}\nonumber 
h'(w) &=& \frac{g'(w)}{g(w)}, \\\label{temp3}
h''(w) &=& \frac{g''(w)}{g(w)} - \frac{(g'(w))^2}{g(w)^2}, \\\nonumber 
h'''(w) &=& 
\frac{g'''(w)}{g(w)} - \frac{3g'(w)g''(w)}{g(w)^2} + \frac{2(g'(w))^3}{g(w)^3}. 
\end{eqnarray}
We can show by induction the following. 
\begin{lemma}\label{lem_h_diff} \ For each $n$ and $w\in \hat{\mathcal {O}}_{\varepsilon  ,+}$, 
\begin{eqnarray*}
h^{(n)}(w) = \sum ^{m_n}_{k = 1}\left( \frac{a_k}{g(w)^{b_k}}\prod ^n_{i=0}(g^{(c_{i,k})}(w))^{d_{i,k}}\right) 
\end{eqnarray*}
for some $m_n, a_k, b_k, c_{i, k}$ and $d_{i, k}$ with $\sum ^n_{i=0}c_{i_k}d_{i_k} = n$. 
\end{lemma}

By (\ref {temp2}), Lemmas \ref {lem_SP_conti} and \ref {lem_h_diff}, 
it suffices to consider the estimation of the order of $g^{(m)}(\hat{w}_\varepsilon )$ for $m\in \Bbb {N}$. 
The next proposition gives the order estimate of $g'(\hat{w}_\varepsilon )$.

\begin{proposition}\label{prop_g_diff_1} \ 
$g'(\hat{w}_\varepsilon ) = O(\varepsilon )$ as $\varepsilon \rightarrow 0$. 
\end{proposition}

\begin{proof} 
By Lemma \ref {lem_g_diff}, we have 
\begin{eqnarray*}
wg'(w) = \theta '(w) - g(w) = \frac{w\theta '(w) - \theta (w)}{w}. 
\end{eqnarray*}
Combining this with (\ref {theta_diff}), we get 
\begin{eqnarray}\label{temp_1_0}
wg'(w) = \frac{w(w-\hat{w}_\varepsilon ) - \theta (w)(K'(\theta (w)) - x)}{w(K'(\theta (w)) - x)}. 
\end{eqnarray}
Letting $w\rightarrow \hat{w}_\varepsilon $, 
both the numerator and the denominator of the right-hand side of (\ref {temp_1_0}) converge to zero. 
Then, we can apply l'H\^opital's rule to obtain 
\begin{eqnarray}\nonumber 
\lim _{w\rightarrow \hat{w}_\varepsilon }wg'(w) &=& 
\lim _{w\rightarrow \hat{w}_\varepsilon }
\frac{2w - \hat{w}_\varepsilon  - \theta '(w)(K'(\theta (w)) - x) - \theta (w)K''(\theta (w))\theta '(w)}
{K'(\theta (w)) - x + wK''(\theta (w))\theta '(w)}\\
&=& 
\frac{\hat{w}_\varepsilon - \hat{\theta }_\varepsilon K''(\hat{\theta }_\varepsilon )\hat{\theta }'_\varepsilon }
{\hat{w}_\varepsilon K''(\hat{\theta }_\varepsilon )\hat{\theta }'_\varepsilon } . 
\label{temp_1_1}
\end{eqnarray}
By Proposition \ref {first_diff_theta} and (\ref {temp_1_1}), we see that 
$g'(\hat{w}_\varepsilon ) = \lim _{w \rightarrow \hat{w}_\varepsilon }g'(w)$ exists and can be given as 
\begin{eqnarray}\label{temp_1_3}
g'(\hat{w}_\varepsilon )  = 
\frac{\hat{w}_\varepsilon  - \sqrt{K''(\hat{\theta }_\varepsilon )}\hat{\theta }_\varepsilon }{\hat{w}_\varepsilon ^2\sqrt{K''(\hat{\theta }_\varepsilon )}}. 
\end{eqnarray}
Our assertion follows from (\ref {temp_1_3}) and Proposition \ref {diff_w_theta}. 
\end{proof}

Differentiating both sides of (\ref {theta_diff}) with respect to $w$, we get the following proposition. 

\begin{proposition} \ \label{prop_2_diff_theta} For 
$w\in \hat{\mathcal {O}}_{\varepsilon  ,+}\setminus \{\hat{w}_\varepsilon \}$, 
\begin{eqnarray}\label{2_diff_theta}
\theta ''(w) =  
\frac{1 - (\theta '(w))^2K''(\theta (w))}{K'(\theta (w)) - x}. 
\end{eqnarray}
\end{proposition}

By (\ref {temp_1_0}) and Propositions \ref {lem_g_diff} and \ref {prop_2_diff_theta}, we obtain the following.

\begin{proposition}\label{prop_2_diff_g} \ For 
$w\in \hat{\mathcal {O}}_{\varepsilon  ,+}\setminus \{\hat{w}_\varepsilon \}$, 
\begin{eqnarray*}
g''(w) = 
\frac{w^2(1-(\theta ')^2K''(\theta )) - 2(K'(\theta ) - x)(w\theta ' - \theta )}{w^3(K'(\theta ) - x)}, 
\end{eqnarray*}
with $\theta = \theta (w)$ and $\theta ' = \theta '(w)$ for brevity. 
\end{proposition}

Next, we consider the second derivative $\hat{\theta }''_\varepsilon = \hat{\theta }^{(2)}_\varepsilon $ 
of $\theta (w)$ at $\hat{w}_\varepsilon $. 

\begin{proposition}\label{second_diff_theta} \ 
\begin{eqnarray}\label{hat_theta_2}
\hat{\theta }''_\varepsilon = -\frac{K'''(\hat{\theta }_\varepsilon )}{3(K''(\hat{\theta }_\varepsilon ))^2}. 
\end{eqnarray}
\end{proposition}

\begin{proof} 
Apply l'H\^opital's rule for (\ref {2_diff_theta}) and observe that
\begin{eqnarray*}
\hat{\theta }''_\varepsilon  
&=& 
-\lim _{w\rightarrow \hat{w}_\varepsilon }
\frac{2\theta '(w)\theta ''(w)K''(\theta (w)) +(\theta '(w))^3K'''(\theta (w))}{K''(\theta (w))\theta '(w)}\\
&=& 
-2\hat{\theta }_\varepsilon ''
-\frac{K'''(\hat{\theta }_\varepsilon )}{(K''(\hat{\theta }_\varepsilon ))^2}. 
\end{eqnarray*}
We then obtain our assertion by solving the above equation for $\hat{\theta }''_\varepsilon $. 
\end{proof}

\begin{proposition} \ \label{prop_g2}
$g''(\hat{w}_\varepsilon ) = O(\varepsilon ^2)$ as $\varepsilon \rightarrow 0$. 
\end{proposition}

\begin{proof} 
Applying l'H\^opital's rule for the equality in Proposition \ref {prop_2_diff_g} 
and using Proposition \ref {second_diff_theta}, we have 
\begin{eqnarray}
\lim _{w\rightarrow \hat{w}_\varepsilon }wg''(w)
= 
\frac{-\hat{w}_\varepsilon ^2K'''(\hat{\theta }_\varepsilon ) - 
6(K''(\hat{\theta }_\varepsilon ))^{3/2}(\hat{w}_\varepsilon  - 
\hat{\theta }_\varepsilon \sqrt{K''(\hat{\theta }_\varepsilon )})}{3\hat{w}_\varepsilon ^2(K''(\hat{\theta }_\varepsilon ))^2}. 
\label{temp_1_10}
\end{eqnarray}
Similarly to Proposition \ref {diff_w_theta}, by applying Taylor's theorem, we get 
\begin{eqnarray}\label{temp_w_hat_2}
\hat{w}^2_\varepsilon - K''(\hat{\theta }_\varepsilon )\hat{\theta }_\varepsilon ^2 + 
\frac{1}{3}K'''(\hat{\theta }_\varepsilon )\hat{\theta }_\varepsilon ^3 = 
\hat{\theta }_\varepsilon ^2K''(\hat{\theta }_\varepsilon ) v_\varepsilon , 
\end{eqnarray}
where 
\begin{eqnarray*}
v_\varepsilon = -\frac{\hat{\theta }_\varepsilon }{3K''(\hat{\theta }_\varepsilon )}
\int ^1_0K^{(4)}(-u\hat{\theta }_\varepsilon )(1-u)^3du. 
\end{eqnarray*}
Note that $v_\varepsilon  = O(\varepsilon ^2)$ as $\varepsilon \rightarrow 0$ by [A5]. 
From (\ref {temp_w_hat_2}), we get 
\begin{eqnarray*}
\hat{w}_\varepsilon = 
\hat{\theta }_\varepsilon \sqrt{K''(\hat{\theta }_\varepsilon )}
\sqrt{1 - \frac{K'''(\hat{\theta }_\varepsilon )}{3K''(\hat{\theta }_\varepsilon )}\hat{\theta }_\varepsilon + v_\varepsilon }. 
\end{eqnarray*}
Therefore, 
we can rewrite the numerator of the right-hand side of (\ref {temp_1_10}) as 
\begin{eqnarray*}
&&-\left \{ K''(\hat{\theta }_\varepsilon )\hat{\theta }_\varepsilon ^2 - 
\frac{1}{3}K'''(\hat{\theta }_\varepsilon )\hat{\theta }_\varepsilon ^3 + 
\hat{\theta }_\varepsilon ^2K''(\hat{\theta }_\varepsilon )v_\varepsilon \right \} K'''(\hat{\theta }_\varepsilon )\\&& - 
6(K''(\hat{\theta }_\varepsilon ))^2
\hat{\theta }_\varepsilon \left\{ 
\sqrt{1 - \frac{K'''(\hat{\theta }_\varepsilon )}{3K''(\hat{\theta }_\varepsilon )}\hat{\theta }_\varepsilon + v_\varepsilon } - 1 \right \} \\
&=& 
-\left \{ K''(\hat{\theta }_\varepsilon )\hat{\theta }_\varepsilon ^2 - 
\frac{1}{3}K'''(\hat{\theta }_\varepsilon )\hat{\theta }_\varepsilon ^3 \right \} 
K'''(\hat{\theta }_\varepsilon ) + 
K''(\hat{\theta }_\varepsilon )\hat{\theta }_\varepsilon ^2K'''(\hat{\theta }_\varepsilon ) + O(\varepsilon ^2)\\
&=& 
\frac{1}{3}(K'''(\hat{\theta }_\varepsilon ))^2\hat{\theta }_\varepsilon ^3 + O(\varepsilon ^2)\ = \ O(\varepsilon ^2) \ \ \mathrm {as} 
\ \ \varepsilon \rightarrow 0. 
\end{eqnarray*}
Here, we use the relations 
$\sqrt{1+x} = 1 + x/2 + O(x^2)$ for small $x$, 
$K'''(\hat{\theta }_\varepsilon ) = O(\varepsilon )$, and 
$v_\varepsilon  = O(\varepsilon ^2)$ as $\varepsilon \rightarrow 0$. 
This completes the proof. 
\end{proof}

In fact, we can refine the assertion of the above proposition. 
From Taylor's theorem, we observe that 
\begin{eqnarray*}
\hat{w}^2_\varepsilon - K''(\hat{\theta }_\varepsilon )\hat{\theta }_\varepsilon ^2 + 
\frac{1}{3}K'''(\hat{\theta }_\varepsilon )\hat{\theta }_\varepsilon ^3 - 
\frac{1}{12}K^{(4)}(\hat{\theta }_\varepsilon )\hat{\theta }^4_\varepsilon  = 
\hat{\theta }_\varepsilon ^2K''(\hat{\theta }_\varepsilon )\tilde{v}_\varepsilon 
\end{eqnarray*}
to arrive at 
\begin{eqnarray*}
\hat{w}_\varepsilon = \hat{\theta }_\varepsilon \sqrt{K''(\hat{\theta }_\varepsilon )}
\sqrt{ 1 - \frac{K'''(\hat{\theta }_\varepsilon )}{3K''(\hat{\theta }_\varepsilon )}\hat{\theta }_\varepsilon + 
\frac{K^{(4)}(\hat{\theta }_\varepsilon )}{12K''(\hat{\theta }_\varepsilon )}\hat{\theta }_\varepsilon ^2 + 
\tilde{v}_\varepsilon }, 
\end{eqnarray*}
where 
\begin{eqnarray*}
\tilde {v}_\varepsilon = 
\frac{\hat{\theta }_\varepsilon ^2}{12K''(\hat{\theta }_\varepsilon )}
\int ^1_0K^{(5)}(-u\hat{\theta }_\varepsilon )(1-u)^4du \ \left(  = O(\varepsilon ^3)\ \ \mathrm {as} \ \ \varepsilon \rightarrow 0\right) . 
\end{eqnarray*}
Then, 
by a calculation similar to that in the proof of the above proposition, we have 
\begin{eqnarray*}
&&3\hat{w}^2_\varepsilon (K''(\hat{\theta }_\varepsilon ))^2\lim _{w\rightarrow \hat{w}_\varepsilon }wg''(w)\\
&=& 
-\left \{ K''(\hat{\theta }_\varepsilon )\hat{\theta }_\varepsilon ^2 - 
\frac{1}{3}K'''(\hat{\theta }_\varepsilon )\hat{\theta }_\varepsilon ^3 + 
\frac{1}{12}K^{(4)}(\hat{\theta }_\varepsilon )\hat{\theta }_\varepsilon ^4 + 
\hat{\theta }_\varepsilon ^2K''(\hat{\theta }_\varepsilon )\tilde{v}_\varepsilon \right \} K'''(\hat{\theta }_\varepsilon )\\&& - 
6(K''(\hat{\theta }_\varepsilon ))^2
\hat{\theta }_\varepsilon \left\{ 
\sqrt{1 - \frac{K'''(\hat{\theta }_\varepsilon )}{3K''(\hat{\theta }_\varepsilon )}\hat{\theta }_\varepsilon + 
\frac{K^{(4)}(\hat{\theta }_\varepsilon )}{12K''(\hat{\theta }_\varepsilon )}\hat{\theta }_\varepsilon ^2 + 
\tilde{v}_\varepsilon } - 1 \right \} \\
&=& 
-\left \{ K''(\hat{\theta }_\varepsilon )\hat{\theta }_\varepsilon ^2 - 
\frac{1}{3}K'''(\hat{\theta }_\varepsilon )\hat{\theta }_\varepsilon ^3\right \} 
K'''(\hat{\theta }_\varepsilon )\\&& + 
K''(\hat{\theta }_\varepsilon )\hat{\theta }_\varepsilon ^2K'''(\hat{\theta }_\varepsilon ) - 
\frac{1}{4}K''(\hat{\theta }_\varepsilon )\hat{\theta }_\varepsilon ^3K^{(4)}(\hat{\theta }_\varepsilon ) + 
\frac{3}{4}(K''(\hat{\theta }_\varepsilon ))^2
\hat{\theta }_\varepsilon \left( \frac{K'''(\hat{\theta }_\varepsilon )}{3K''(\hat{\theta }_\varepsilon )}\hat{\theta }_\varepsilon \right)^2 + O(\varepsilon ^3)\\
&=& 
\frac{5}{12}\hat{\theta }_\varepsilon ^3(K'''(\hat{\theta }_\varepsilon ))^2 - 
\frac{1}{4}K''(\hat{\theta }_\varepsilon )\hat{\theta }_\varepsilon ^3K^{(4)}(\hat{\theta }_\varepsilon ) + O(\varepsilon ^3) \ \ \mathrm {as} 
\ \ \varepsilon \rightarrow 0, 
\end{eqnarray*}
where we have applied the relation $\sqrt{1 + x} = 1 + x/2 - x^2/8 + O(x^3)$ for small $x$. 
This implies that 
\begin{eqnarray}\label{temp_g_2_diff}
g''(\hat{w}_\varepsilon ) = 
\left( \frac{5(K'''(\hat{\theta }_\varepsilon ))^2}{36(K''(\hat{\theta }_\varepsilon ))^2} - 
\frac{K^{(4)}(\hat{\theta }_\varepsilon )}{12K''(\hat{\theta }_\varepsilon )}
\right)\frac{\hat{\theta }_\varepsilon ^3}{\hat{w}_\varepsilon ^3} + O(\varepsilon ^3) \ \ \mathrm {as} 
\ \ \varepsilon \rightarrow 0. 
\end{eqnarray}

Here, we calculate the third derivative of $\theta (w)$ at $\hat{w}_\varepsilon $ 
($\hat{\theta }'''_\varepsilon $). 

\begin{proposition}\label{prop_3_diff_theta_eps} \ 
\begin{eqnarray}\label{3_diff_theta_eps}
\hat{\theta }'''_\varepsilon  = 
\frac{5(K'''(\hat{\theta }_\varepsilon ))^2}{12(K''(\hat{\theta }_\varepsilon ))^{7/2}} - 
\frac{K^{(4)}(\hat{\theta }_\varepsilon )}{4(K''(\hat{\theta }_\varepsilon ))^{5/2}}. 
\end{eqnarray}
\end{proposition}

\begin{proof} 
Differentiating both sides of (\ref {2_diff_theta}), we have 
\begin{eqnarray}
\theta '''(w) 
&=& 
-\frac{3\theta '\theta ''K''(\theta ) + (\theta ')^3K'''(\theta )}{K'(\theta ) - x}. 
\label{3_diff_theta}
\end{eqnarray}
Now we apply l'H\^opital's rule for (\ref {3_diff_theta}) to obtain 
\begin{eqnarray*}
\hat{\theta }'''_\varepsilon 
= \lim _{w\rightarrow \hat{w}_\varepsilon }\theta '''(w) = 
-\left\{ -\frac{(K'''(\hat{\theta }_\varepsilon ))^2}{3(K''(\hat{\theta }_\varepsilon ))^{7/2}} + 
3\hat{\theta }'''_\varepsilon  - 
\frac{2(K'''(\hat{\theta }_\varepsilon ))^2}{(K''(\hat{\theta }_\varepsilon ))^{7/2}} + 
\frac{K^{(4)}(\hat{\theta }_\varepsilon )}{(K''(\hat{\theta }_\varepsilon ))^{5/2}}\right\}. 
\end{eqnarray*}
This can be simplified to 
\begin{eqnarray*}
4\hat{\theta }'''_\varepsilon  = 
\frac{5(K'''(\hat{\theta }_\varepsilon ))^2}{3(K''(\hat{\theta }_\varepsilon ))^{7/2}} - 
\frac{K^{(4)}(\hat{\theta }_\varepsilon )}{(K''(\hat{\theta }_\varepsilon ))^{5/2}}. 
\end{eqnarray*}
We have obtained the desired assertion. 
\end{proof}

Substituting (\ref {3_diff_theta_eps}) into (\ref {temp_g_2_diff}), we have the following proposition. 
\begin{proposition} \ \label{prop_g2_further}
\begin{eqnarray}\label{g2_further}
g''(\hat{w}_\varepsilon ) = 
\frac{(\hat{\theta }_\varepsilon \sqrt{K''(\hat{\theta }_\varepsilon )})^3}{3\hat{w}^3_\varepsilon }\times 
\hat{\theta }'''_\varepsilon  + O(\varepsilon ^3), \ \ \varepsilon \rightarrow 0. 
\end{eqnarray}
\end{proposition}

Now we are prepared to prove the next proposition. 

\begin{proposition} \ \label{prop_g3}
$g'''(\hat{w}_\varepsilon ) = O(\varepsilon ^3)$ as $\varepsilon \rightarrow 0$. 
\end{proposition}

\begin{proof} 
By Lemma \ref {lem_g_diff}, it holds that 
\begin{eqnarray*}
wg'''(w) = \theta '''(w) - 3g''(w)
\end{eqnarray*}
for $w\neq \hat{w}_\varepsilon $. 
Letting $w\rightarrow \hat{w}_\varepsilon $ and substituting (\ref {g2_further}), we have 
\begin{eqnarray*}
\hat{w}_\varepsilon \lim _{w\rightarrow \hat{w}_\varepsilon }g'''(w) &=& 
\hat{\theta }'''_\varepsilon  - 3\left\{ 
\frac{(\hat{\theta }_\varepsilon \sqrt{K''(\hat{\theta }_\varepsilon )})^3}{3\hat{w}^3_\varepsilon }\times 
\hat{\theta }'''_\varepsilon  + O(\varepsilon ^3)\right\} \\
&=& 
\frac{\hat{\theta }'''_\varepsilon }{\hat{w}^3_\varepsilon }
\left\{ \hat{w}^3_\varepsilon  - (\hat{\theta }_\varepsilon \sqrt{K''(\hat{\theta }_\varepsilon )})^3\right\}  + O(\varepsilon ^3). 
\end{eqnarray*}
By [A5] and Proposition \ref {prop_3_diff_theta_eps}, we see that 
$\hat{\theta }'''_\varepsilon  = O(\varepsilon ^2)$. 
Moreover, Proposition \ref {diff_w_theta} implies that 
\begin{eqnarray*}
&&\hat{w}^3_\varepsilon  - (\hat{\theta }_\varepsilon \sqrt{K''(\hat{\theta }_\varepsilon )})^3\\ 
&=& 
\left( \hat{w}_\varepsilon  - \hat{\theta }_\varepsilon \sqrt{K''(\hat{\theta }_\varepsilon )}\right) 
\left( \hat{w}_\varepsilon ^2 + 
\hat{w}_\varepsilon \hat{\theta }_\varepsilon \sqrt{K''(\hat{\theta }_\varepsilon )} + 
\hat{\theta }_\varepsilon ^2K''(\hat{\theta }_\varepsilon ) \right) = O(\varepsilon ), 
\ \ \varepsilon \rightarrow 0. 
\end{eqnarray*}
By the above arguments, we deduce that 
$\hat{w}_\varepsilon g''(\hat{w}_\varepsilon ) = O(\varepsilon ^3)$ as $\varepsilon \rightarrow 0$. 
\end{proof}

Next we estimate $\hat{\theta }^{(n)}_\varepsilon $ and $g^{(n)}(\hat{w}_\varepsilon )$ 
for $n\geq 4$. 
We let 
\begin{eqnarray*}
f_n(w) = \theta ^{(n)}(w)(K'(\theta (w)) - x).
\end{eqnarray*} 

\begin{lemma} \ \label{rel_f} $f_{n+1}(w) = f'_n(w) - K''(\theta (w))\theta '(w)\theta ^{(n)}(w)$ 
for each $n\geq 1$. 
\end{lemma}

\begin{proof} 
A straightforward calculation gives 
\begin{eqnarray*}
\theta ^{(n+1)} = \frac{d}{dw}\left( \frac{f_n}{K'(\theta ) - x}\right) = 
\frac{f'_n\cdot (K'(\theta ) - x) - f_n\cdot K''(\theta )\theta '}{(K'(\theta ) - x)^2} = 
\frac{f'_n - \theta ^{(n)}K''(\theta )\theta '}{K'(\theta ) - x}, 
\end{eqnarray*}
which implies the desired assertion. 
\end{proof}

\begin{proposition} \ \label{estimate_fn}For each $n\geq 3$, the following two assertions hold. \\
$\mathrm {(i)}$ \ There are nonnegative integers $m^n, a^n_i, r^n_i, s^n_i, k^n_{i,2}, \ldots , k^n_{i, n-2}$ $(i = 1, \ldots , m^n)$ such that 
\begin{eqnarray*}
f_n(w) &=& -K^{(n)}(\theta (w))(\theta '(w))^n - nK''(\theta (w))\theta '(w)\theta ^{(n-1)}(w)\\
&& - 
\sum ^{m^n}_{i = 1}a^n_iK^{(r^n_i)}(\theta (w))(\theta '(w))^{s^n_i}\prod ^{n-2}_{j=2}(\theta ^{(j)}(w))^{k^n_{i, j}} 
\end{eqnarray*}
and also $\sum ^{n-2}_{j=2}(j-1)k^n_{i, j} + r^n_i = n$, $r^n_i\geq 2$ for each $i = 1, \ldots , m^n$. \\
$\mathrm {(ii)}$ \ $f_n(\hat{w}_\varepsilon ) = 0$. 
\end{proposition}

\begin{proof} We will prove assertion (i) 
by induction. First, we consider the case $n = 3$. 
By Proposition \ref {prop_2_diff_theta} and Lemma \ref {rel_f}, we know 
\begin{eqnarray*}
f_2(w) = 1 - K''(\theta (w))(\theta '(w))^2
\end{eqnarray*}
and 
\begin{eqnarray}\nonumber 
f_3(w) &=& f_2'(w) - K''(\theta (w))\theta '(w)\theta ''(w)\\ 
&=& 
-K'''(\theta (w))(\theta '(w))^3 - 3K''(\theta (w))\theta '(w)\theta ''(w); 
\label{temp_f3}
\end{eqnarray}
thus, (i) is true for $n = 3$. 

Now we assume that (i) holds for any integer in $\{ 3, \ldots , n \}$. 
Thus, 
\begin{eqnarray*}
f_{n+1}(w) &=& f'_n(w) - K''(\theta )\theta '\theta ^{(n)}\\
&=& 
-K^{(n+1)}(\theta )(\theta ')^{n+1} - (n+1)K''(\theta )\theta '\theta ^{(n)}\\
&& - \Big\{ nK^{(n)}(\theta )(\theta ')^{n-1}\theta '' + nK'''(\theta )(\theta ')^2\theta ^{(n-1)}\\&&\hspace{5mm} + 
nK''(\theta )\theta ''\theta ^{(n-1)} + \sum ^{m^n}_{i=1}a^n_iF^n_i(\theta (w))\Big\} 
\end{eqnarray*}
by virtue of Lemma \ref {rel_f}, where 
\begin{eqnarray*}
F^n_i(\theta ) &=& 
K^{(r^n_i + 1)}(\theta )(\theta ')^{s^n_i + 1}\prod ^{n-2}_{j=2}(\theta ^{(j)})^{k^n_{i, j}} \\
&& + 
s^n_iK^{(r^n_i)}(\theta )(\theta ')^{s^n_i-1}\theta ''\prod ^{n-2}_{j=2}(\theta ^{(j)})^{k^n_{i, j}} \\
&& + 
\sum ^{n-2}_{l=2}
k^n_{i, l}K^{(r^n_i)}(\theta )(\theta ')^{s^n_i}(\theta ^{(l)})^{k^n_{i, j} - 1}\theta ^{(l+1)}
\prod ^{n-2}_{j=2: j\neq l}(\theta ^{(j)})^{k^n_{i, j}}. 
\end{eqnarray*}
Replacing $n$ with $n+1$ again gives (i). By induction, (i) holds for $n \geq 3$.
The assertion (ii) is obvious from (\ref {def_SP}) and the definition of $f_n(w)$. 
\end{proof}

\begin{proposition} \ \label{prop_theta_n}For each $n\geq 2$, we have 
$\hat{\theta }^{(n)}_\varepsilon = O(\varepsilon ^{n-1})$ as $\varepsilon \rightarrow 0$. 
\end{proposition}

\begin{proof} 
When $n=2$, the assertion is obvious by [A5] and Proposition \ref {second_diff_theta}. 
We suppose that the assertion is true for $1, \ldots , n-1$. 
By the definition of $f_n$, we have 
\begin{eqnarray*}
\theta ^{(n)}(w) = \frac{f_n(w)}{K'(\theta (w)) - x} 
\end{eqnarray*}
for $w\neq \hat{w}_\varepsilon $. 
By Proposition \ref {estimate_fn}(ii) and the definition of $\hat{\theta }_\varepsilon $, we see that 
both the numerator and the denominator of the right-hand side of the above equality converge to zero by letting $w\rightarrow \hat{w}_\varepsilon $. 
Therefore, we can apply l'H\^opital's rule to obtain 
\begin{eqnarray}\label{temp_2_1}
\hat{\theta }^{(n)}_\varepsilon = 
\lim _{w\rightarrow \hat{w}_\varepsilon }\frac{f'_n(w)}{K''(w)\theta '(w)} = 
\frac{f'_n(\hat{w}_\varepsilon )}{\sqrt{K''(\hat{\theta }_\varepsilon )}}. 
\end{eqnarray}
By Lemma \ref {rel_f} and Proposition \ref {estimate_fn}, 
we see that $f'_n(\hat{w}_\varepsilon )$ has the form 
\begin{eqnarray}
f'_n(\hat{w}_\varepsilon ) &=& -n\sqrt{K''(\hat{\theta }_\varepsilon )}\hat{\theta }^{(n)}_\varepsilon - 
\sum ^{m^n}_{i = 1}a^n_iK^{(r^n_i)}(\hat{\theta }_\varepsilon )(\hat{\theta }'_\varepsilon )^{s^n_i}\prod ^{n-1}_{j=2}(\hat{\theta }^{(j)}_\varepsilon )^{k^n_{i, j}}\label{temp_2_2}
\end{eqnarray}
for some $m^n, a^n_i, r^n_i, s^n_i, k^n_{i,2}, \ldots , k^n_{i, n-1}$ $(i = 1, \ldots , m^n)$ with $\sum ^{n-1}_{j=2}(j-1)k^n_{i, j} + r^n_i = n+1$. 
By (\ref {temp_2_1})--(\ref {temp_2_2}), we have 
\begin{eqnarray*}
\hat{\theta }^{(n)}_\varepsilon = \frac{1}{(n+1)\sqrt{K''(\hat{\theta }_\varepsilon )}}\sum ^{m^n}_{i = 1}a^n_iK^{(r^n_i)}(\hat{\theta }_\varepsilon )(\hat{\theta }'_\varepsilon )^{s^n_i}\prod ^{n-1}_{j=2}(\hat{\theta }^{(j)}_\varepsilon )^{k^n_{i, j}}. 
\end{eqnarray*}
Here, by the supposition $\hat{\theta }^{(j)}_\varepsilon  = O(\varepsilon ^{j-1})$ as 
$\varepsilon \rightarrow 0$ for $j = 2, \ldots , j = n-1$ and that [A4] holds, 
we see that the term 
\begin{eqnarray*}
K^{(r^n_i)}(\hat{\theta }_\varepsilon )(\hat{\theta }'_\varepsilon )^{s^n_i}\prod ^{n-1}_{j=2}(\hat{\theta }^{(j)}_\varepsilon )^{k^n_{i, j}}
\end{eqnarray*}
has order $O(\varepsilon ^{r^n_i-2+\sum _j(j-1)k^n_{i,j}}) = O(\varepsilon ^{n-1})$ 
as $\varepsilon \rightarrow 0$. 
Thus, $\hat{\theta }^{(n)}_\varepsilon  = O(\varepsilon ^{n-1})$ as $\varepsilon \rightarrow 0$. 
Therefore, the assertion is also true for $n$. Induction completes the proof.
\end{proof}

\begin{lemma}\label{lem_order_g_diff} \ For each $n\geq 3$, $g^{(n)}(\hat{w}_\varepsilon ) = O(\varepsilon ^3)$ as $\varepsilon \rightarrow 0$. 
\end{lemma}

\begin{proof} The assertion is true for $n = 3$ by Proposition \ref {prop_g3}. 
For $n \geq 4$, the assertion is obtained by Lemma \ref {lem_g_diff}, Proposition \ref {prop_theta_n}, and induction. 
\end{proof}

\begin{proof}[Proof of Theorem \ref {th_main}] 
Since $(\phi (\hat{w}_\varepsilon ))_{0\leq \varepsilon \leq 1}$ is bounded, 
it suffices to show that 
$\psi ^{(m)}(\hat{w}_\varepsilon ) =
$\\ $
\allowbreak O(\varepsilon ^{\min \{ 2m + 1, 3 \}})$, $\varepsilon \rightarrow 0$ for $m\geq 0$. 
From (\ref {temp2})--(\ref {temp3}), we have that 
\begin{eqnarray*}
\psi _\varepsilon (\hat{w}_\varepsilon ) = \frac{g'(\hat{w}_\varepsilon )}{g(\hat{w}_\varepsilon )} = O(\varepsilon )\ \ \mathrm {as} \ \ \varepsilon \rightarrow 0
\end{eqnarray*}
by Proposition \ref {prop_g_diff_1} and that 
\begin{eqnarray*}
\psi _\varepsilon ''(\hat{w}_\varepsilon ) = 
\frac{g'''(\hat{w}_\varepsilon )}{g(\hat{w}_\varepsilon )} - 
\frac{g'(\hat{w}_\varepsilon )g''(\hat{w}_\varepsilon )}{g(\hat{w}_\varepsilon )^2} + 
\frac{2(g'(\hat{w}_\varepsilon ))^3}{g(\hat{w}_\varepsilon )^3} = O(\varepsilon ^3)\ \ \mathrm {as} \ \ 
\varepsilon \rightarrow 0
\end{eqnarray*}
by Propositions \ref {prop_g_diff_1}, \ref {prop_g2}, and \ref {prop_g3}. 
For $m\geq 2$, we get the assertion by Lemmas \ref {lem_h_diff} and \ref {lem_order_g_diff}. 
\end{proof}

\section{Extentions}\label{sec_extention}

\subsection{Error Estimates of the Higher Order LR Formulae}\label{sec_error_estimate}

In the beginning of this subsection, 
we introduce the following proposition. 

\begin{proposition}\label{prop_rep_g} \ For each $n$, 
\begin{eqnarray}\label{Taylor_series_g}
g^{(n)}(\hat{w}_\varepsilon ) = \sum ^\infty _{k = n + 1}
\frac{n!}{k!}\hat{\theta }^{(k)}_\varepsilon (-\hat{w}_\varepsilon )^{k - n - 1}. 
\end{eqnarray}
\end{proposition}

\begin{proof} 
Using Lemma \ref {lem_g_diff} and induction, 
we see that $g^{(n)}(\hat{w}_\varepsilon )$ can be represented as 
\begin{eqnarray}\nonumber 
\hat{w}_\varepsilon ^{n + 1}g^{(n)}(\hat{w}_\varepsilon ) &=& 
\sum ^n_{k = 0}(-1)^{n-k}\frac{n!}{k!}\hat{w}_\varepsilon ^k\hat{\theta }_\varepsilon ^{(k)}\\
&=& 
(-1)^nn!\theta (\hat{w}_\varepsilon ) + 
\sum ^n_{k = 1}(-1)^{n-k}\frac{n!}{k!}\hat{w}_\varepsilon ^k\hat{\theta }_\varepsilon ^{(k)}. 
\label{temp_g_1}
\end{eqnarray}
Combining (\ref {temp_g_1}) with the Taylor expansion 
\begin{eqnarray*}\label{temp_g_2}
n!\theta (\hat{w}_\varepsilon ) = 
-n!(\theta (0) - \theta (\hat{w}_\varepsilon )) = 
-\sum ^\infty _{k = 1}\frac{n!}{k!}\hat{\theta }^{(k)}_\varepsilon (-\hat{w}_\varepsilon )^k, 
\end{eqnarray*}
we get the desired assertion. 
\end{proof}

Here, by Proposition \ref {prop_theta_n}, there are positive constants 
$C_n$ (with $n\geq 2$), such that 
\begin{eqnarray}\label{estimate_hat_theta}
|\hat{\theta }^{(n)}_\varepsilon |\leq C_n\varepsilon ^{n - 1}. 
\end{eqnarray}
Therefore, if we assume the further condition [A6] below, 
then the series (\ref {Taylor_series_g}) converges absolutely when $\varepsilon $ is small. 
\begin{itemize}
 \item [ {[A6]} ] There exists $\varepsilon _0 \in (0, 1]$ such that 
\begin{eqnarray*}
\sum ^\infty _{k = 2}\frac{C_k}{k!}\varepsilon _0^k < \infty . 
\end{eqnarray*}
\end{itemize}
Moreover, we obtain the following theorem.
\begin{theorem}\label{th_main2} \ 
Assume $\mathrm {[A1]}$--$\mathrm {[A6]}$. 
Then $h^{(n)}(\hat{w}_\varepsilon ) = O(\varepsilon ^n)$, $\varepsilon \rightarrow 0$ holds for each $n\geq 1$. 
Moreover, $\Psi ^\varepsilon _m(\hat{w}_\varepsilon ) = O(\varepsilon ^{2m + 1})$, $\varepsilon \rightarrow 0$ holds for each $m\geq 0$. 
\end{theorem}

\begin{proof}
This is an immediate consequence of (\ref {Taylor_series_g}) and Lemma \ref {lem_h_diff}. 
\end{proof}

By the above theorem, we see that 
there are positive constants 
$C'_n$ (where $n\geq 2$) such that 
$|h^{(n)}(\hat{w}_\varepsilon )|\leq C'_n\varepsilon ^{n}$, and hence 
\begin{eqnarray*}
|\Psi ^\varepsilon _m(\hat{w}_\varepsilon )| \leq \phi (\hat{w}_\varepsilon )\frac{C'_{2m + 1}}{(2m)!!}\varepsilon ^{2m+1}. 
\end{eqnarray*}
Now we introduce the condition [A7].
\begin{itemize}
 \item [ {[A7]} ] There exists $\varepsilon _1 \in (0, 1]$ such that 
\begin{eqnarray*}
\sum ^\infty _{m = 1}\frac{C'_{2m + 1}}{(2m)!!}\varepsilon _1^{2m+1} < \infty. 
\end{eqnarray*}
\end{itemize}

Then, obviously we have the theorem below.

\begin{theorem}\label{th_main3} \ 
Assume $\mathrm {[A1]}$--$\mathrm {[A7]}$and that $(\ref {exact_LR})$ holds. 
Then the expansion formula $(\ref {error_estimate})$ holds. 
\end{theorem}

Note that [A6]--[A7] are technical conditions that may be hard to verify directly in the general case. 
However, the results in Section \ref {sec_eg} suggest that 
the assertions of Theorems \ref {th_main2}--\ref {th_main3} are likely to be valid 
in many cases. 

\subsection{Application to the Daniels Formula for Density Functions}
\label{sec_Daniels}

In this subsection, 
we study the order estimates for the saddlepoint approximation formula of Daniels (1954), 
which approximates the probability density function. 
Let $x\in \Bbb {R}$ and define $\hat{\theta }^{(n)}_\varepsilon , \hat{w}_\varepsilon $ 
as are done in Section \ref {sec_proof_estimate}. 
By an argument similar to that in Section \ref {sec_LR}, 
we can prove the following ``exact'' Daniels expansion: 
\begin{eqnarray}\label{exact_Daniels}
f_\varepsilon (x) = \sum ^\infty _{m = 0}\Theta _m 
\end{eqnarray}
under suitable conditions, 
where $f_\varepsilon $ is the probability density function of $\mu _\varepsilon $ and 
\begin{eqnarray*}
\Theta _m= \phi (\hat{w}_\varepsilon )\frac{\hat{\theta }^{(2m + 1)}_\varepsilon }{(2m)!!}. 
\end{eqnarray*}
In the case of the sample mean of i.i.d.\hspace{1mm}random variables, 
this version of (\ref {exact_Daniels}) was studied as (3.3) in Daniels (1954) and 
(2.5) in Daniels (1980). 
In the general case, 
we can obtain (\ref {exact_Daniels}) 
under, for instance, [A1]--[A5], [B1]--[B2] and the following additional condition. 
\begin{itemize}
 \item [ {[A8]}] There exists  $\varepsilon _2\in (0, 1]$ such that 
\begin{eqnarray*}
\sum ^{\infty }_{n = 1}\frac{C_n}{n!!}\varepsilon _2^n < \infty , 
\end{eqnarray*}
where $C_n > 0$ is a constant appearing in (\ref {estimate_hat_theta}). 
\end{itemize}

We can easily show the following by arguments similar to those in Section 
\ref {sec_proofs} and Subsection \ref {sec_error_estimate} (we omit the proof here). 
\begin{theorem}\label{th_Daniels} \ 
Assume $\mathrm {[A1]}$--$\mathrm {[A5]}$. 
Moreover assume that $(\ref {exact_Daniels})$ holds. 
Then $\Theta _m = O(\varepsilon ^{2m})$ as $\varepsilon \rightarrow 0$ for each $m\geq 0$. 
Moreover, if we further assume $\mathrm {[A8]}$, 
it holds that 
\begin{eqnarray*}
f_\varepsilon (x) = \sum ^M_{m = 0}\Theta _m + O(\varepsilon ^{2(M+1)})\ \ \mathrm {as} 
\ \ \varepsilon \rightarrow 0\ \ \mathrm {for \ each} \ \ M\geq 0. 
\end{eqnarray*}
\end{theorem}

\section{Concluding Remarks}\label{sec_conclusion}

For a general, parametrised sequence of
random variables $(X^{(\varepsilon)})_{\varepsilon>0}$, 
assuming that the $r$th cumultant of $X^{(\varepsilon)}$
has order $O(\varepsilon ^{r-2})$ 
as $\varepsilon \rightarrow 0$ for each $r\geq 3$, 
we derive the ``exact'' Lugannnani-Rice expansion formula
for the right tail probability
\[
P\left(X^{(\varepsilon)}>x \right) 
= 1 - \Phi (\hat{w}_\varepsilon ) 
+ \sum ^\infty _{m = 0}\Psi ^\varepsilon _m(\hat{w}_\varepsilon ),
\]
where $x\in {\mathbb R}$ is fixed to a given value. 
In particular, we have obtained 
the order estimates of each term in the expansion. 
For the first two terms, we have that 
$\Psi ^\varepsilon _0(\hat{w}_\varepsilon )=O(\varepsilon )$ and 
$\Psi ^\varepsilon _1(\hat{w}_\varepsilon )=O(\varepsilon ^3)$
as $\varepsilon \to 0$, respectively. 
Under some additional conditions, 
the $m$th term satisfies
$\Psi ^\varepsilon _m(\hat{w}_\varepsilon ) = O(\varepsilon ^{2m + 1})$ 
as $\varepsilon \to 0$. 
Using these, we have established (\ref {error_estimate}) 
for each $m, M\geq 0$. 
As numerical examples, 
we chose stochastic volatility models in financial mathematics;
we checked the validity of our order estimates for the LR formula.

The following are interesting and important 
future research topics related to this work.
\begin{itemize}
 \item[(i)]
Analysing the far-right tail probability 
\begin{eqnarray*}
P\left( X^{(\varepsilon)}>\frac{x}{\varepsilon}\right), 
\end{eqnarray*}
using an LR type expansion,  
which is compatible with the classical LR formula 
(see Remark \ref {rem_intro} in Introduction).
In this case, the saddlepoint diverges as $\varepsilon \rightarrow 0$ 
allowing us to avoid the difficulty in calculating
(\ref {integral_sum}) 
by using Watson's lemma (see Watson (1918) or Kolassa (1997)). 
Hence, we can expect that condition [B3] may be omitted;
this condition was imposed when we derived the exact LR expansion.

 \item[(ii)]
Seeking more ``natural'' conditions 
than [A6]--[A7] for obtaining the error estimate (\ref {error_estimate}).

 \item[(iii)]
Studying order estimates for 
generalized LR expansions with non-Gaussian bases. 
Among studies of the expansions without order estimates are  
Wood, Booth and Butler (1993), Rogers and Zane (1999), 
Butler (2007), and Carr and Madan (2009).
\end{itemize}

\appendix

\section{Explicit Forms of Higher Order Approximation Terms}
\label{sec_higher_order}
In this section, we introduce the derivation of $\Psi ^\varepsilon _2(\hat{w}_\varepsilon )$ and 
$\Psi ^\varepsilon _3(\hat{w}_\varepsilon )$. 
First, we can inductively calculate $\hat{\theta }^{(r)}_\varepsilon $ for $r\geq 4$ by 
the same calculation as the proof of Proposition \ref {prop_theta_n}. 
\begin{proposition} \ \label{theta_high}
\begin{eqnarray*}
\hat{\theta }^{(4)}_\varepsilon &=& 
-\frac{K^{(5)}(\hat{\theta }_\varepsilon )}{5(K''(\hat{\theta }_\varepsilon ))^3} + 
\frac{K^{(3)}(\hat{\theta }_\varepsilon )K^{(4)}(\hat{\theta }_\varepsilon )}{(K''(\hat{\theta }_\varepsilon ))^4} - 
\frac{8(K^{(3)}(\hat{\theta }_\varepsilon ))^3}{9(K''(\hat{\theta }_\varepsilon ))^5}, \\
\hat{\theta }^{(5)}_\varepsilon &=& 
-\frac{K^{(6)}(\hat{\theta }_\varepsilon )}{6(K''(\hat{\theta }_\varepsilon ))^{7/2}} + 
\frac{35 (K^{(4)}(\hat{\theta }_\varepsilon ))^2}{48(K''(\hat{\theta }_\varepsilon ))^{9/2}} + 
\frac{7K^{(3)}(\hat{\theta }_\varepsilon )K^{(5)}(\hat{\theta }_\varepsilon )}{6(K''(\hat{\theta }_\varepsilon ))^{9/2}}\\
&& - 
\frac{35(K^{(3)}(\hat{\theta }_\varepsilon ))^2K^{(4)}(\hat{\theta }_\varepsilon )}{8(K''(\hat{\theta }_\varepsilon ))^{11/2}} + 
\frac{385(K^{(3)}(\hat{\theta }_\varepsilon ))^4}{144(K''(\hat{\theta }_\varepsilon ))^{13/2}}, \\
\hat{\theta }^{(6)}_\varepsilon &=& 
-\frac{K^{(7)}(\hat{\theta }_\varepsilon )}{7(K''(\hat{\theta }_\varepsilon ))^{4}} - 
\frac{280(K^{(3)}(\hat{\theta }_\varepsilon ))^5}{27(K''(\hat{\theta }_\varepsilon ))^{8}} + 
\frac{200(K^{(3)}(\hat{\theta }_\varepsilon ))^3K^{(4)}(\hat{\theta }_\varepsilon )}
{9(K''(\hat{\theta }_\varepsilon ))^{7}}\\
&& - 
\frac{25(K^{(4)}(\hat{\theta }_\varepsilon ))^2}{3(K''(\hat{\theta }_\varepsilon ))^{6}} - 
\frac{20(K^{(3)}(\hat{\theta }_\varepsilon ))^2K^{(5)}(\hat{\theta }_\varepsilon )}
{3(K''(\hat{\theta }_\varepsilon ))^{6}} + 
\frac{2K^{(4)}(\hat{\theta }_\varepsilon )K^{(5)}(\hat{\theta }_\varepsilon )}
{(K''(\hat{\theta }_\varepsilon ))^{5}}\\&& + 
\frac{4K^{(3)}(\hat{\theta }_\varepsilon )K^{(6)}(\hat{\theta }_\varepsilon )}
{3(K''(\hat{\theta }_\varepsilon ))^{5}}, \\
\hat{\theta }^{(7)}_\varepsilon &=& 
-\frac{K^{(8)}(\hat{\theta }_\varepsilon )}{8(K''(\hat{\theta }_\varepsilon ))^{9/2}} - 
\frac{85085(K^{(3)}(\hat{\theta }_\varepsilon ))^6}{1728(K''(\hat{\theta }_\varepsilon ))^{19/2}} - 
\frac{25025(K^{(3)}(\hat{\theta }_\varepsilon ))^4K^{(4)}(\hat{\theta }_\varepsilon )}
{192(K''(\hat{\theta }_\varepsilon ))^{17/2}}\\
&& + 
\frac{5005(K^{(3)}(\hat{\theta }_\varepsilon ))^2(K^{(4)}(\hat{\theta }_\varepsilon ))^2}
{64(K''(\hat{\theta }_\varepsilon ))^{15/2}} - 
\frac{385(K^{(4)}(\hat{\theta }_\varepsilon ))^3}
{64(K''(\hat{\theta }_\varepsilon ))^{13/2}} + 
\frac{1001(K^{(3)}(\hat{\theta }_\varepsilon ))^3K^{(5)}(\hat{\theta }_\varepsilon )}
{24(K''(\hat{\theta }_\varepsilon ))^{15/2}}\\
&& - 
\frac{231K^{(3)}(\hat{\theta }_\varepsilon )K^{(4)}(\hat{\theta }_\varepsilon )K^{(5)}(\hat{\theta }_\varepsilon )}
{8(K''(\hat{\theta }_\varepsilon ))^{13/2}} + 
\frac{63(K^{(5)}(\hat{\theta }_\varepsilon ))^2}
{40(K''(\hat{\theta }_\varepsilon ))^{11/2}} - 
\frac{77(K^{(3)}(\hat{\theta }_\varepsilon ))^2K^{(6)}(\hat{\theta }_\varepsilon )}
{8(K''(\hat{\theta }_\varepsilon ))^{13/2}}\\
&& + 
\frac{21K^{(4)}(\hat{\theta }_\varepsilon )K^{(6)}(\hat{\theta }_\varepsilon )}
{8(K''(\hat{\theta }_\varepsilon ))^{11/2}} + 
\frac{3K^{(3)}(\hat{\theta }_\varepsilon )K^{(7)}(\hat{\theta }_\varepsilon )}
{2(K''(\hat{\theta }_\varepsilon ))^{11/2}}. 
\end{eqnarray*}
\end{proposition}

Second, 
by continuing the differentiation in (\ref {temp3}), we have 
\begin{eqnarray*}
h^{(4)}(w) &=& 
\frac{g^{(4)}(w)}{g(w)} - 
\frac{6(g'(w))^4}{g(w)^4} + 
\frac{12(g'(w))^2g''(w)}{g(w)^2} - 
\frac{3(g''(w))^2}{g(w)^2} - 
\frac{4g'(w)g^{(3)}(w)}{g(w)^2}, \\
h^{(5)}(w) &=& 
\frac{g^{(5)}(w)}{g(w)} - 
\frac{24(g'(w))^5}{g(w)^5} - \frac{60(g'(w))^3g''(w)}{g(w)^4} + 
\frac{30g'(w)(g''(w))^2}{g(w)^3}\\
&& + 
\frac{20(g'(w))^2g^{(3)}(w)}{g(w)^2} - 
\frac{10g''(w)g^{(3)}(w)}{g(w)^2} - 
\frac{5g'(w)g^{(4)}(w)}{g(w)^2}, \\
h^{(6)}(w) &=& 
\frac{g^{(6)}(w)}{g(w)} - 
\frac{120(g'(w))^6}{g(w)^6} + 
\frac{360(g'(w))^4g''(w)}{g(w)^5} - 
\frac{270(g'(w))^2(g''(w))^2}{g(w)^4}\\
&& + 
\frac{30(g''(w))^3}{g(w)^2} - 
\frac{120(g'(w))^3g^{(3)}(w)}{g(w)^4} + 
\frac{120g'(w)g''(w)g^{(3)}(w)}{g(w)^3}\\
&& - 
\frac{10(g^{(3)}(w))^2}{g(w)^2} + 
\frac{30(g'(w))^2g^{(4)}(w)}{g(w)^3} - 
\frac{15g''(w)g^{(4)}(w)}{g(w)^2} - 
\frac{6g'(w)g^{(5)}(w)}{g(w)^2}, \\
h^{(7)}(w) &=& 
\frac{g^{(7)}(w)}{g(w)} + 
\frac{720(g'(w))^7}{g(w)^7} - 
\frac{2520(g'(w))^5g''(w)}{g(w)^6} + 
\frac{2520(g'(w))^3(g''(w))^2}{g(w)^5}\\
&& - 
\frac{630g'(w)(g''(w))^3}{g(w)^4} + 
\frac{840(g'(w))^4g^{(3)}(w)}{g(w)^5} - 
\frac{1260(g'(w))^2g''(w)g^{(4)}(w)}{g(w)^4}\\
&& + 
\frac{210(g''(w))^2g^{(3)}(w)}{g(w)^2} + 
\frac{140g'(w)(g^{(3)}(w))^2}{g(w)^3} - 
\frac{210(g'(w))^3g^{(4)}(w)}{g(w)^4}\\
&& + 
\frac{210g'(w)g''(w)g^{(4)}(w)}{g(w)^3} - 
\frac{35g^{(3)}(w)g^{(4)}(w)}{g(w)^2} + 
\frac{42(g'(w))^2g^{(5)}(w)}{g(w)^3}\\
&& - 
\frac{21g''(w)g^{(5)}(w)}{g(w)^2} - 
\frac{7g'(w)g^{(6)}(w)}{g(w)^2}, 
\end{eqnarray*}
where $g(w)$ and $h(w)$ are defined as (\ref {def_g}). 
Combining this with (\ref {temp2}), Lemma \ref {lem_g_diff}, and Propositions \ref {first_diff_theta}, \ref {second_diff_theta}, \ref {prop_3_diff_theta_eps}, and \ref {theta_high}, 
we can calculate $\Psi ^\varepsilon _2(\hat{w}_\varepsilon )$ and $\Psi ^\varepsilon _3(\hat{w}_\varepsilon )$ explicitly.

\section*{Acknowledgement}

The authors thank communications with 
Masaaki Fukasawa of Osaka University, 
who directed their attentions to the Lugannani--Rice formula.
Jun Sekine's research was supported by a Grant-in-Aid
for Scientific Research (C), No.\hspace{1mm}23540133,
from the Ministry of Education, Culture,
Sports, Science, and Technology, Japan.

\end{document}